\documentclass{article} 
\usepackage{geometry}
\usepackage{latexsym} 
\usepackage{amsfonts} 
\usepackage{amsmath} 
\usepackage{amsthm} 
\usepackage{mathrsfs} 
\usepackage{dsfont} 
\usepackage{bbold} 
\usepackage[english]{babel} 
\usepackage{caption} 
\usepackage{epsfig} 
\usepackage{float} 
\usepackage{subfigure} 
\usepackage{psfrag} 
\usepackage{graphicx,xcolor} 
\usepackage{epsfig} 
\usepackage{amsbsy} 
\usepackage{amssymb}
\usepackage{color} 
\usepackage[colorlinks=true, linkcolor=magenta, citecolor=cyan, urlcolor=blue]{hyperref}

\DeclareMathOperator{\Val}{\matV}

 \DeclareMathOperator{\sign}{sign}

\DeclareMathOperator{\loc}{loc} 

\newtheorem*{theorem}{Theorem}
 
\theoremstyle{plain}
\newtheorem{theorema}{Theorem}[section]
\newtheorem{lemma}[theorema]{Lemma}
\newtheorem{prop}[theorema]{Proposition}

\theoremstyle{remark}
\newtheorem{rmk}[theorema]{Remark}

 \definecolor{light}{gray}{.9}

\newcommand{\beq}{\begin{equation}}
\newcommand{\eeq}{\end{equation}}
\def\bea{\begin{eqnarray}}
\def\eea{\end{eqnarray}}

\newcommand{\ZZZ}{\mathds{Z}}

\newcommand{\RRR}{\mathds{R}} 
\newcommand{\TTT}{\mathds{T}} 
\newcommand{\uno}{\mathds{1}} 
 
\newcommand{\calF}{{\mathcal F}} 
\newcommand{\calG}{{\mathcal G}}

\newcommand{\calJ}{{\mathcal J}}

\newcommand{\RR}{{\mathcal R}}

\newcommand{\matF}{{\mathscr F}} 
 
\newcommand{\matH}{{\mathscr H}}

\newcommand{\ol}{\overline} 
 
\newcommand{\io}{\infty} 
\newcommand{\e}{\varepsilon} 
\newcommand{\al}{\alpha} 
\newcommand{\de}{\delta} 
\newcommand{\be}{\beta}

\newcommand{\x}{\xi} 
 
\newcommand{\ka}{\kappa} 
\newcommand{\g}{\gamma} 
\newcommand{\om}{\omega}

\newcommand{\f}{\varphi} 
\newcommand{\s}{\sigma} 
 
\newcommand{\del}{\partial}

\newcommand{\nn}{\boldsymbol{n}}

\newcommand{\ee}{\boldsymbol{e}}

\newcommand{\ii}{{\rm i}}

\def\tilde#1{\widetilde{#1}}
 
\def\ins#1#2#3{\vbox to0pt{\kern-#2 \hbox{\kern#1 #3}\vss}\nointerlineskip}

\numberwithin{equation}{section}

\def\be{\begin{equation}}
\def\ee{\end{equation}}
\def\bea{\begin{eqnarray}}
\def\eea{\end{eqnarray}}
\def\nn{\nonumber}

\def\sx{\sigma^x}
\def\sy{\sigma^y}
\def\sz{\sigma^z}

\def\a{\alpha}
\def\d{\delta}
\def\g{\gamma}

\def\D{\Delta}

\def\r{\rho}
\def\u{\upsilon}
\def\t{\tau}
\def\dg{\dagger}
\def\II{\mathds{1}}
\def\w{\omega}

\def\sign{\operatorname{sign}}
\def\Z{\mathds{Z}}
\def\T{\mathds{T}}
\def\f{\varphi}
\def\hp{\psi}

\def\R{\mathds{R}}
\def\N{\mathds{N}}

\newcommand{\ie}{i.e.}

\newcommand{\OOO}[1]{O \left(#1\right)}
\newcommand{\OO}[1]{O \left(\frac{1}{#1}\right)}

\def\spect{\operatorname{spect}}
\def\pv{\operatorname{p.v.}}
\def\sinc{\operatorname{sinc}}
\def\cosc{\operatorname{cosc}}

\newcommand{\nocontentsline}[3]{}
\newcommand{\tocless}[2]{\bgroup\let\addcontentsline=\nocontentsline#1{#2}\egroup}

\DeclareMathSymbol{\leqslant}{\mathalpha}{AMSa}{"36} 
\DeclareMathSymbol{\geqslant}{\mathalpha}{AMSa}{"3E} 
\DeclareMathSymbol{\eset}{\mathalpha}{AMSb}{"3F}     
\renewcommand{\leq}{\;\leqslant\;}                   
\renewcommand{\geq}{\;\geqslant\;}                   
\newcommand{\dd}{\,\text{\rm d}}

\def\Val{\operatorname{Val}}

\begin{document}
 
\title{\bf Periodic Driving at High Frequencies of an Impurity in the Isotropic XY Chain}
 
\author 
{ Livia Corsi$^{1}$ and Giuseppe Genovese$^{2}$
\vspace{2mm} 
\\ \small
$^{1}$ School of Mathematics, Georgia Institute of Technology, 686 Cherry St. NW, Atlanta GA, 30332, USA
\\ \small 
$^{2}$ Institut f\"ur Mathematik, Universit\"at Z\"urich
Winterthurerstrasse 190, CH-8057 Z\"urich, CH
\\ \small 
E-mail:  lcorsi6@math.gatech.edu, giuseppe.genovese@math.uzh.ch}

\date{\today} 
 
\maketitle 
\begin{abstract}
We study the isotropic XY chain with a transverse magnetic field acting on a single site and analyse the long time behaviour of the time-dependent state of the system when a periodic perturbation drives the impurity. We find that for high frequencies the state approaches a periodic orbit synchronised with the forcing and provide the explicit rate of convergence to the asymptotics.  

\end{abstract}


\section{Introduction}\label{intro} 

Frequently time-dependent Hamiltonians provide a fair effective description of otherwise complicated physical systems. In particular the interaction of a quantum system with a classical electromagnetic field typifies non-autonomous quantum mechanical models. Time-periodic perturbations correspond to monochromatic radiation and constitute the basic case to analyse. 

In this paper we study the isotropic XY
 quantum spin chain with a periodically time-dependent transverse external field acting only in one site, namely the $\kappa$-th. The Hamiltonian operator reads
\be\label{eq:HXX-Vt}
H_N(t)=-g\sum_{j=1}^{N-1}\left(\sx_j\sx_{j+1}+\sy_{j}\sy_{j+1}\right)-(V_0+hV(\w t))\sz_\kappa\,.
\ee

Here $ 1<\kappa<N$ is fixed once and for all, $\sx,\sy,\sz$ denote the Pauli matrices, $g,h,\w>0$ are three parameters ruling respectively the strength of the spin-spin coupling, the magnitude of the external field and its frequency. We assume that $V(\om t)$ is a real analytic function with period $2\pi/\om$ and zero mean, \ie
\be\label{eq:V}
V(\om t)=\sum_{\substack{n\in\ZZZ}} e^{\ii n \om t} V_n\,,
\qquad |V_n|\le C_0 e^{-\s |n|}\,.
\ee

For any time $t$ the Hamiltonian is a self-adjoint operator on the tensor product of $1/2$-spin vector spaces $\mathcal{H}_j=\mathbb{C}^2$, each spanned by the vectors \textit{spin up} and \textit{spin down}: $\mathcal{H}_N:=\bigotimes_{j=1}^{N} \mathcal{H}_j={\mathbb{C}^2}^{\otimes N}$. The corresponding matrix algebra of $2\times 2$ matrices $GL^2(\mathbb{C})$ is spanned by the Pauli matrices $\sx,\sy,\sz$ plus the identity $\uno$. As usual the thermodynamic limit for the system is  performed in the Fock space, defined as $\mathcal{F}:=\bigoplus_N \mathcal{H}_N$. We assume for simplicity free boundary conditions.

Our main result is that, if the frequency is sufficiently high, the state of the system approaches a periodic orbit with the same period of the perturbation. It will be stated precisely below. 

The analysis of this model was begun in \cite{ABGM0, ABGM1, ABGM2}. The Hamiltonian (\ref{eq:HXX-Vt}) can be conveniently written in terms of quasi-free fermions via a Jordan-Wigner transformation, \ie
$$
\left\{
\begin{array}{lll}
\sx_j&=&(c_j^\dg+c_j)\bigotimes_{k=1}^{j-1}(-\sz_k)\\
\sy_j&=&-i(c_j^\dg-c_j)\bigotimes_{k=1}^{j-1}(-\sz_k)\\
\sz_j&=&2c^\dg_jc_j-\uno\nn
\end{array}
\right.
$$
followed by a Fourier transform (from $\Z\cap[-N,N]$ to the $N$-th cyclotomic group $\Z(N)$). It becomes
\be\label{eq:Hqp}
H_N(t)=-\sum_{q,p\in\Z(N)}\left[g\cos p\d_{qp}+(V_0+hV(\w t))\frac1Ne^{i\kappa(q-p)}\right]a^\dg_qa_p
\ee
where $\de_{qp}$ is the Kronecker delta and $a^\dg, a$ satisfy the canonical anti-commutation relations 
$$
\left[a_p,a_q\right]_{+}=0,\quad\left[a^\dg_p,a_q\right]_{+}=\d_{pq}\,.
$$
The equilibrium property of such a system are well understood, as the N-body state is determined by the single particle. 
The Hamiltonian was diagonalised in \cite{ABGM0, ABGM1}. Let $\bar h$ be the magnitude of the field at fixed time and
$$
\mathcal{H}^N_{q,p}:=g\cos q\d_{qp}+\frac{\bar h}{N}e^{i(q-p)\kappa}\,,
$$
be the Laplacian on $\Z$ with a rank-one perturbation. $H_N$ is the second quantisation of $\mathcal{H}^N$. Then as $N\to\infty$ 
\be\label{eq:spectHn}
\spect[\mathcal{H}_N]\longrightarrow[-g,g]\cup\left\{g\sign(\bar h)\sqrt{1+\frac{\bar h^2}{g^2}}\right\}\,.
\ee

As we will see below, the dynamical properties of the model are still determined by the one particle state, but the action of the rest of the chain adds an additional (linear) term in the one-particle Schr\"odinger equation. 

Our work stems from the analysis performed in \cite{ABGM2}. There the authors investigated the motion of the impurity with various time-dependent external fields. In particular they computed the magnetisation of the perturbed spin at the first order in $h$, with $V(\w t)=\cos\w t$, observing a resonance ({\ie} a divergence) in $\w= 2g$. 

We will see that such divergences appear in any order of the (formal) expansion in $h$. They must be controlled by an appropriate renormalisation. We can already foresee where the problem is by looking at the spectrum. For $V_0\neq0$ we are essentially perturbing the isolated eigenvalue. Note that the spectral gap between the band and the eigenvalue is approximately $\frac{V_0}{g}(V_0+hV(\w t))$. Therefore if $h$ is small enough we have no crossing of eigenvalues, hence no resonances.
When $V_0=0$ the perturbation could move energy levels within the band. However, if $\w$ is sufficiently high we are able to treat the resonances.

In the rest of this introduction we will introduce our set up, state the main result and link it to some literature on the topic.

\subsection{One-particle reduction}

Now we briefly show how to reduce the dynamics of the $N$-body problem to a one-body Schr\"odinger equation with memory, adopting the scheme of \cite{ABGM2} (for more details, we refer to \cite{GG}).

The starting point is the representation (\ref{eq:Hqp}). Obviously for $t'\neq t''$ one has $[H_N(t'),H_N(t'')]\neq0$, which implies a non-trivial dynamics. At fixed $N$, the solution of the Heisenberg equations
\be\label{eq:Heisenberg}
\frac{d}{dt} a_q(t)=[a_q(t),-iH(t)], \quad q\in\mathds{Z}(N).
\ee
is given by a one-parameter semigroup of Bogoliubov transformations, namely
\be\label{eq:B-V1par}
a_q(t)=\frac1N\sum_p A_{qp}(t)a_p\,,
\ee
where the matrix $A_{qp}(t)$ solves the following ODE (obtained by matching (\ref{eq:Heisenberg}) and (\ref{eq:B-V1par}))
$$
i\dot A_{qp}(t)=\frac1N\sum_{q'}A_{qq'}(t)\left[g\cos p\d_{pq'}+(V_0+hV(\w t))\frac1Ne^{i\ka(q'-p)}\right]\,,\quad \,A_{qp}(0)=\d_{qp}\,.
$$
Fix $t_0\in\RRR$ as the initial time;
then the following change of variables will be convenient:
$$
\Psi_{j,q}(t):=e^{ig\cos q(t-t_0)}\frac1N\sum_{p\in\Z(N)}e^{ipj} A_{qp}(t)\,,\quad j\in\Z\cap[-N/2,N/2]\,.
$$
Obviously, the $A_{qp}$ are equivalent to the $\Psi_{j,q}$. Moreover, as we are dealing with quasi-free fermions, the functions $\{\Psi_{j,q}(t)\}_{j\in\Z}$ completely specify the state of the system at any time $t$.
For any finite $N$, the $\{\Psi_{j,q}(t)\}_{j\in\Z}$ for all times are trivially given by the standard existence and uniqueness theorem for ODEs. As $N\to\infty$, after some manipulations and using Duhamel formula we arrive at the following set of equations
\be\label{eq:Psi-j}
\Psi_j(q;t)=1-ie^{iq(j-\kappa)}\int_{t_0}^t dt' J_{j-\kappa}(g(t-t')) (V_0+hV(\w t))\Psi_{\kappa}(q;t')\,,\quad j\in\Z\,,q\in\T:=\RRR/(2\pi\ZZZ)\,.
\ee
For $k\in\Z$, $J_k(t)$ denotes the Bessel function of first kind and $k$-th order:
$$
J_{k}(t):=\frac{1}{2\pi}\int_{-\pi}^{\pi}dx e^{ixk+i\cos x t}\,.
$$
In order to be as self-consistent as possible, many properties of Bessel functions used in the paper are proved in Appendix \ref{App:Bessel}.

Since the state of the system is completely determined by $\Psi_{\kappa}(q;t)$, henceforth
we agree to refer to the solution of the following integral equation (with the change of variable $\xi:=\cos q$)
$$
\hat{\psi}(\x,t):=\Psi_{\kappa}(q(\x);t)
$$
and
\be\label{equazione}
\hat{\psi}(\xi,t)=1-i\int_{t_0}^t\dd t' J_0(g(t-t'))e^{ig\x(t-t')}(V_0+hV(\w t'))\hat{\psi}(\x,t')\,,\quad\xi\in[-1,1]\,
\ee
as the state of the chain. This integral equation comes from the Schr\"odinger equation on $\Z$ of the Floquet type
\begin{equation}\label{THE EQ}
i\del_t\psi(x,t)=-g\Delta\psi(x,t)+H_F(t,t_0)\psi(x,t)\,,\quad \psi(x,0)=\d(x)\,,
\end{equation}
passing to the Fourier side on the spatial variables.
Here $\Delta$ is the Laplacian on $\Z$ and $H_F(t,t_0)$ is
\bea
H_F(t,t_0)\psi(x,t)&:=&(V_0+ h V(\om t))\psi(x,t)-gV_0\int_{t_0}^{t}\dd t' J_1(g(t-t'))e^{\ii\Delta (t-t')}\psi(x, t')\nn\\
&-&gh\int_{t_0}^{t}\dd t' J_1(g(t-t'))e^{\ii\Delta (t-t')}V(\w t')\psi(x, t')\nn\,. 
\eea
Therefore the many-body problem reduces to a one-body Schr\"odinger equation with a memory term. Since the Duhamel representation in the momentum space is much simpler, we will only deal with it for the rest of the paper. Remarkably, this kind of equations appeared already in the original studies by Volterra \cite{Volterra}, who called them heredity equation. They model quite a vast variety of physical phenomena, from elasticity theory to electromagnetism. 

\medskip

Let us spend few words on the notations. With a little abuse of notation, since throughout the paper we operate in the Fourier (spatial) space, we will systematically omit $\hat \cdot$ to indicate spatial Fourier transformations. Therefore we convey that $\psi(\xi,t)$
for $[-1,1]\ni \xi=\cos q$, $q\in[-\pi,\pi]$, is the Fourier coefficient of $\psi(x,t)$,  $x\in\Z$.
The Fourier and Hilbert transform $\matF$ and $\matH$ will be employed. According to the standard definitions
\be\label{fouhil}
\matF[f](\t):=\int_\R\frac{dt}{\sqrt{2\pi}}e^{i\t t}f(t)\,,\quad\matH[f](t):=\frac1\pi\pv\int_\R d t' \frac{f(t-t')}{t'}\,,
\ee
where $\pv$ denotes the Cauchy principal value. $\chi(\cdot)$ always indicates the indicator function of a set (subset of $\R$). We write $X\lesssim Y$ when there is a constant $C$ such that $X\leq CY$.

\subsection{Main result}

It is convenient to introduce the Volterra type operator for any $t>t_0$ and $\xi\in[-1,1]$
\be\label{eq:Wt0}
W_{t_0}f(\x,t):=\int_{t_0}^t\dd t' J_0(g(t-t'))e^{ig\x(t-t')}(V_0+hV(\w t'))f(\x, t')\,.
\ee
This can be regarded as a linear map from $L^{2}_\xi C_t([-1,1]\times[t_0,t])$ into itself. For any $t_0$ finite it is an integral 
operator with smooth kernel on a finite interval, hence compact. We rewrite (\ref{equazione}) as
\be\label{duat0}
(\uno+iW_{t_0})\psi(\x,t)=1\,.
\ee
Therefore by the standard Volterra's theory \cite{Volterra} (for a modern exposition, see for instance \cite{EN}) we obtain the 
existence of a unique solution for finite time, {\ie} as $t-t_0<\infty$. We denote this one-parameter family of functions with $\psi_{t_0}(\xi,t)$. 

The operator $W_{t_0}$ has a formal limit as $t_0\to-\infty$, which we denote by $W_{\infty}$, defined through
\be\label{eq:Winf}
W_{\infty}f(\x,t):=\int_{-\infty}^t\dd t' J_0(g(t-t'))e^{ig\x(t-t')}(V_0+hV(\w t'))f(\x, t')\,.
\ee
$W_{\infty}$ is an unbounded operator on $L^{2}_\x C_t([-1,1]\times \RRR)$, mapping periodic functions of
 frequency $\w$ into periodic functions of frequency $\w$. Thus the formal limit of (\ref{duat0}) is 
\be\label{duat}
(\uno+iW_{\infty})\psi(\x,t)=1\,.
\ee
We will show that this equation has a periodic solution with frequency $\w$, under some conditions on $g,\w,h$ and 
denote such a solution by $\psi_\infty(\xi,t)$. We will work in two different regimes, namely
\vskip.5truecm
\noindent
{\bf (a) }{\it The non-degenerate case}: $V_0\ne0$.

\vskip.5truecm

\noindent
{\bf (b) }{\it Moderately resonant regime}: $V_0=0$ and there is $\e>0$, such that
\begin{equation}\label{perturb}
\frac g\w\le\frac{1}{2}-\e\,. 
\end{equation}

Our main result is the convergence as $t_0\to-\infty$ of the family $\{\psi_{t_0}(\xi,t)\}_{t_0\in\R}$ to a function $\psi_\infty(\xi,\w t)$, which is a periodic solution of (\ref{duat}) with frequency $\w$. More precisely we have
\begin{theorem}\label{main-as}
Let $h\ll \w$ and assume either {\bf (a)} or {\bf (b)}. Then there exists
a periodic solution of (\ref{duat}) with frequency $\om$ analytic in $h$, whose Fourier transform
$\psi_\infty(x,\w t)\in L_x^2C_t^{\infty}(\Z\times\R)$. Moreover, let $\psi_{t_0}(x,t)$
 be the unique solution of (\ref{THE EQ}). It holds 
\be\nonumber
\psi_{t_0}(x,t)=\psi_{\infty}(x,\w t)+\OO{\sqrt{t-t_0}}\,.
\ee
\end{theorem}

\subsection{Discussion and summary}

This class of problems was intensely studied in the past. Two approaches were mainly adopted, both having classical counterparts.
In one of them, introduced by Howland and Yajima in \cite{yaj1, H0}, time is promoted as an additional space coordinate
$\t\in[0,t]$ (as in the contact geometry of the extended phase space with time and energy in classical mechanics.). This permits to study directly the spectral properties of the Floquet operator
$$
i\frac{d}{d\t}-H(\t)\,
$$
in $L^2[0,t]\times \mathcal H$, where $\mathcal H$ is the Hilbert space of the system at fixed time. Therefore one can develop the theory in parallel with the one for periodic Sch\"odinger operators. This idea has been exploited in many other papers, as for instance \cite{matt, bru, bach}, dealing with open quantum systems, in a set up similar to ours. The second technique, first used by Bellissard in \cite{bell}, is based on KAM algorithm and have been much developed in the past decades in different contexts, and in particular for time-dependent Schr\"odinger equations (see e.g. \cite{BG, EK, woj2}).

Here we are using none of these methods, as our analysis is chiefly based on a renormalisation of the Neumann series. Of course there is much overlap among these three approaches. Nonetheless we cannot use KAM reducibility for our problem. The technical reason is that, in the more difficult
case of zero-average forcing term, we deal with perturbations of the identity (bear in mind (\ref{duat})). This trivialises the homological equation at each KAM step.
Moreover we emphasise that all the aforementioned papers (except \cite{EK}) deal with perturbation of the discrete spectrum, while we have the Laplacian on $\Z$. This constitutes a major issue to cope with in our work. 

A striking feature of the solution we find is that even if the external field has zero average, the system approaches a time-periodic state with non-zero average (see Remark \ref{rmk:mean}). The lack of ergodicity of this sort of models is of much interest in non-equilibrium statistical physics and it received attention by old and new works \cite{ABGM2}\cite{Leb}\cite{rob}\cite{mastro}. In any case, since we deal with high frequency, we are not able to establish whether or not the limits $t_0\to-\infty$ and $\w\to0$ commutes. So the question remains in the realm of conjectures.

The high-frequency assumption appears also in other works on related problems \cite{woj1, woj2, BDP, bach}. From a mathematical viewpoint the main issue at low frequencies is that, as $\w$ is smaller than the width of the continuous spectrum, a finite but arbitrary large number of harmonics are resonant with the chain \cite{Gius}. The study of this more delicate regime requires some new idea.

\medskip

The paper is organised as follows. 
In Sections \ref{eq} and \ref{sect:finite-time} we prove the main theorem. The strategy of our proof consists of two steps. First in Section \ref{eq} we find a periodic solution of (\ref{duat}), analytic in $h$. This is achieved by expanding the resolvent of the operator $W_{\infty}$ in powers of $h/\w$. 
The case $V_0\neq0$  is simpler. When the potential has zero average, the oscillating eigenvalue crosses the band, giving rise to resonances. We take care of these by an appropriate renormalisation. This part of our analysis is much in the spirit of trees formalism, first developed for KAM theory in \cite{E}. 
The linearity of the problem allows us to deal with linear trees, which we call reeds. 
In Section \ref{sect:finite-time} we show that the solution of (\ref{duat0}) approaches the periodic solution constructed in Section 2 as $t_0\to-\infty$. We give the explicit rate of convergence, which is determined by the asymptotic behaviour of the Bessel function $J_0$. The non-degenerate case is easier to treat and the proof, long but direct, is only sketched. The moderately resonant regime is more delicate and is presented in detail. Indeed in that case the resonances give two possible bad effects that we have to avoid: first, as in Section 2, the single terms of the expansions can be divergent; second even if regularised, the resonances could in principle kill the desired decay in time, which depends indeed on the time oscillations. 
We attach two appendices, whose results are heavily employed in Section \ref{sect:finite-time}. In Appendix \ref{App:Bessel} many useful oscillatory integrals are explicitly computed; in Appendix \ref{B} we present a separate proof of the main theorem as $h=0$.

\paragraph*{Acknowledgements} We thank G. Gallavotti for suggesting the problem and G. Gentile, A. Giuliani and B. Schlein for some stimulating discussions. This paper has been written as L.C. was Canada Research Chair Postdoctoral fellow at McMaster University. L.C. was partly supported by
the NSF grant DMS-1500943. G.G. was supported by the Swiss National Science Foundation through the grant ``Effective equations from quantum dynamics".


\section{The asymptotic solution} 
\label{eq} 

Now we concentrate on the solution of (\ref{duat}). It will be helpful to set 
$$
\a:=\frac g\w\,,\qquad\g:=\frac h\w\,,
$$
although will use this notation only in this section. We want to prove the following result.
\begin{prop}\label{asintoticsol}
Let $\g$ small enough, and either {\bf (a)} or {\bf (b)} hold. Then there exists
a periodic solution to \eqref{THE EQ} with frequency $\om$ analytic in $\g$.
\end{prop}

\subsection{Formal set up}

Let us consider the change of variables $\f:=\w t$, so that equation (\ref{THE EQ}) (in spatial Fourier variables) becomes
\begin{equation}\label{eq.forma1}
\begin{aligned}
i\del_\f\hp(\x,\f)=&-\a \x\hp(\x,\f)+(V_0+ \g V(\f))\hp(\x,\f)\\
&-\a\int_{-\infty}^{\f}\dd \f' J_1\left(\frac g\w(\f-\f')\right)e^{\ii\frac \x\w (\f-\f')}(V_0+\g V(\f'))\hp(\x,\f')\,,
\end{aligned}
\end{equation}
and (\ref{duat}) becomes
\begin{equation}\label{dua}
\hp(\x,\f)=1-\ii\int_{-\infty}^\f \dd \f' J_0\left(\a(\f-\f')\right)e^{\ii\a \x (\f-\f')}(V_0+\g V(\f'))\hp(\x,\f')\,.
\end{equation}
Set
\be\label{j}
j_k=j(k+\a \x):=\frac{\chi(\{\x\,:\,|k+\a \x|\le\a\})-\ii\chi(\{\x\,:\,|k+\a \x|>\a\})}{\sqrt{|(\a \x+k)^2-\a^2|}}
\ee
and note that each $j_k$ is singular for $\a \x+k\pm\a=0$.

\begin{rmk}\label{verij}
Once $\a$ is fixed there is $\ol{\e}>0$ such that, for all $k$ with $|k|\ge2\a+\ol\e$ one has
\begin{equation}\label{jbello}
|j_k(\x)|\le \frac{1}{\sqrt{2\a\ol\e+\ol\e^2}}\,,
\end{equation}
uniformly in $\x$.
In particular, in the {\em moderately resonant} regime \eqref{perturb}
the only divergence is for $k=0$, as we have
\be\label{j0}
j_0(\x)=\frac{1}{\sqrt{\a^2( 1-\x^2)}}
\ee
while
\be\label{jk}
|j_k(\x)| =\left| \frac{-\ii}{\sqrt{(\a \x+k)^2-\a^2}}\right|\le \frac{1}{\sqrt{2\ol{\e}}},\qquad k\ne0\,.
\ee
\end{rmk}

Let us now compute the asymptotic solution as a formal power series in $\g$
\begin{equation}\label{formalsol}
\psi = \sum_{k\ge0}\g^k \psi_{k}\,,
\end{equation}
where
\be
\begin{array}{lll}\label{eq:y_n}
\psi_0&:=&1\,,\\
\psi_k&:=&W_\infty[\psi_{k-1}]\,.
\end{array}
\ee
Since the image via $W_\infty$ of a periodic function is periodic, each term of the series is a periodic function of $\f$,
so we can consider its Fourier expansion also in time. Using the fact that
\be\nonumber
(W_\infty u)_k=j_k\sum_{n\in\ZZZ}V_{k-n}u_n\,,\quad (\psi_0)_k=\d_{0,k}\,,
\ee
by direct computation we obtain
\be\label{eq:yk}
\left\{
\begin{aligned}
\psi_1(\f)&=\sum_{n_1\in \ZZZ}j_{n_1}V_{n_1}e^{\ii n_1\f}\,,\\
 &\vdots\\
\psi_k(\f)&=\sum_{n_1,\dots,n_k\in \ZZZ}\Big(\prod_{i=1}^k  j_{\mu_i} V_{n_i}\Big)e^{\ii\mu_k\f}\,,
\end{aligned}
\right.
\ee
where we denoted
\be\label{conserva}
\mu_i=\mu(n_1,\ldots,n_i):=\sum_{j=1}^i n_j\,.
\ee
In this way we obtained that the formal power series
\be\label{formale}
\begin{aligned}
\tilde \psi(\x,\f;\g):&=\sum_{\mu\in\ZZZ}e^{\ii\mu\f}\psi_\mu(\x;\g)
=\sum_{\mu\in\ZZZ}e^{\ii\mu\f} \sum_{N\geq0}(-i\g)^N
\psi_{N,\mu}(\x)\\
&=\sum_{\mu\in\ZZZ}e^{\ii\mu\f}\sum_{N\geq0}
\sum_{\substack{n_1,\ldots,n_N\in \ZZZ\\ \mu_N=\mu}}
(-i\g)^{N}
\Big(\prod_{i=1}^N j_{\mu_i}(\x) V_{n_i}\Big)\,,
\end{aligned}
\ee
solves \eqref{THE EQ} at each order $\g$.
Thus the uniform convergence of the series \eqref{formale} for $\g$ small enough will provide a well defined solution of \eqref{THE EQ}. To achieve this goal we shall present a graphical representation for the coefficients $\psi_{N,\mu}$
in \eqref{formale} and use such representation as a tool to simplify the computations leading to the
appropriate bounds, which in turn entail the convergence of the series. 

\subsection{Reeds expansion} 
\label{canne} 

Here we introduce the graphical formalism which will allow us to deduce the convergence
of the formal series \eqref{formale}; such formalism has been used plenty of times
both in theoretical and mathematical physics.
 In the context of KAM
theory it was originally introduced by Gallavotti in \cite{Gal1}, inspired by a
pioneering result obtained by Eliasson in \cite{El0} (published in \cite{E}) and thereafter has been used in many
other related papers; see \cite{G10} for a review.

Since our problem is linear, we deal with linear trees, or \emph{reeds}.

An oriented tree is a graph with no cycle, such that all the lines are oriented toward a single point (the \emph{root})
which has only one incident line (called \emph{root line}). All the points in a tree except the root are called
\emph{nodes}. A \emph{reed} is a linear tree, i.e. a tree in which each node has  exactly two incident lines.
Note that in a reed the orientation induces a natural total ordering ($\preceq$) on the set of the nodes $N(\rho)$.
Moreover, since a line $\ell$ may be identified by the node $v$ which it exits, we have a natural total ordering
also on the set of lines $L(\rho)$.

Given a reed $\rho$ we associate labels with each node and line as follows.
We associate with each node $v$ a \emph{mode label} $n_v\in \ZZZ$ and with each line 
$\ell$ a {\it momentum}
$\mu_\ell \in \ZZZ$ with the constraint
\begin{equation}\label{conservareed}
\mu_\ell = \sum_{v\prec \ell} n_v\,.
\end{equation}
Note that \eqref{conservareed} above is a reformulation of the constraint \eqref{conserva}.
We call \emph{order} of a reed $\rho$ the number of nodes with non-zero mode in it and \emph{total momentum} of a reed the momentum associated with the root line.
$\Theta_{N,\mu}$ denotes the set of reeds of order $N$ and total momentum $\mu$.

We then associate with each node $v$ a \emph{node factor}
\be\label{nodefactor}
\calF_v = V_{n_v}
\ee
and with each line $\ell$ a \emph{propoagator}
\be\label{propagator}
\calG_\ell = j_{\mu_\ell}(\x)\,,
\ee
so that we can associate with each reed $\rho$ a value as
\be\label{val}
\Val(\rho) = \Big(\prod_{v\in N(\rho)} \calF_v\Big) \Big(\prod_{\ell\in L(\rho)} \calG_\ell\Big).
\ee
For any $\rho\in \Theta_{N,\mu}$ we can rewrite \eqref{val} as  
\be\label{valord}
\Val(\rho)=\prod_{i=1}^N \calF_{v_i}\calG_{\ell_i}\,,
\ee
where $v_i$, $\ell_i$ are the $i$-th node and line respectively. Clearly
\be\label{ovvio}
\psi_{N,\mu} = \sum_{\rho\in \Theta_{N,\mu}} \Val(\rho) \,.
\ee

\subsection{The renormalisation procedure}
\label{renorm}

We pass now to illustrate the renormalisation separately for the two regimes.

\subsubsection{The non-degenerate case}

By Remark \ref{verij} the only dangerous reeds are those containing a line $\ell$ with
momentum $|\mu_\ell|\le 2\a$ since otherwise the
propagators are easily bounded: we call \emph{singular lines} such lines and \emph{regular} the others.
Since the average of the potential is not zero we may have two  lines carrying the same momentum
attached to the same node (having zero mode by the conservation law \eqref{conservareed}): in this
case we say that the two  lines are \emph{connected}. Let $\Theta^{\RR}_{N,\mu}$ denote the set
of {\it renormalised reeds} i.e. reeds in which no pair of connected lines appear. We define
\begin{equation}\label{rino1}
j_\mu^{\RR}(\x):= \frac{j_{\mu}(\x)}{1-\ii V_0 j_\mu(\x)} 
\end{equation}
and for any $\r\in\Theta^\RR_{N,\mu}$ let us define the renormalised value of $\rho$ as
\be\label{renval1}
\Val^\RR(\rho):=\prod_{i=1}^N \calF_{v_i}\calG^{\RR}_{\ell_i}\,,
\ee
where
\be\label{renprop1}
\calG^{\RR}_{\ell_i}= j_\mu^{\RR}(\x)\,.
\ee

We easily see that formally
$$
j_\mu^{\RR}(\x)= j_\mu(\x) \sum_{p\ge0}(-\ii V_{0}{j}_\mu(\x))^p\,,
$$
and since for a renormalised reed the order
coincides with the number of nodes (and hence lines),
setting
\be\label{coeffrin1}
\psi_\mu^\RR(\x;\g):=\sum_{N\ge1}(-\ii\g)^N
\sum_{\r\in\Theta^\RR_{N,\mu}}\Val^\RR(\rho)
\ee
formally we have
$$
\psi_\mu(\x;\g)=\psi_\mu^\RR(\x;\g)\,.
$$
The advantage is that for the propagators \eqref{rino1} we have the bound
\be\label{real}
|j_\mu^{\RR}(\x)|\le
\left\{
\begin{aligned}
&\frac{1}{\sqrt{2\a\ol\e+\ol\e^2}}\qquad &|\mu|\ge2\a+\ol\e,\\
&| V_0|^{-1}\qquad &|\mu|< 2\a+\ol\e.
 \end{aligned}
 \right.
\ee

Therefore
\begin{equation}\label{mestavoascorda}
\begin{aligned}
|\Val^\RR(\r)|&=\Big(\prod_{v\in N(\rho)} |\calF_v|\Big) \Big(\prod_{\ell\in L(\rho)}| \calG_\ell|\Big)
\le \Big( C_0e^{-\s\sum_{v\in N(\rho)}|n_v|}\Big) \Big(\prod_{\ell\in L(\rho)}| \calG_\ell|\Big)\\
&\le 
C_0^NC_1^{-N}e^{-\s\sum_{v\in N(\rho)}|n_v|}\le
(C_0C_1^{-1})^{N}e^{-\s|\mu|}\,,
\end{aligned}
\end{equation}
which implies
\be\label{stimatotale}
|\psi_\mu^\RR(\x;\g)|\le\sum_{N\ge1}\g^{ N} (C_0C_1^{-1})^{N}e^{-\s|\mu|/2}
\ee
and we see that there is a constant $C_1>0$ such that the series above converges for
\begin{equation}\label{gammanon}
\g\le C_0^{-1}C_1\,.
\end{equation}

This means that the function
\begin{equation}\label{asy}
\psi^\RR(\f;\x,\g):=
\sum_{\mu\in\ZZZ}e^{\ii\mu\f}\psi_\mu^\RR(\x;\g)\,,
\end{equation}
is well defined, uniformly bounded for $\x\in[-1,1]$ and $\g$ satisfying \eqref{gammanon}
and analytic for $\f\in\TTT$.

\subsubsection{The moderately resonant  regime}

Anew by Remark \ref{verij} the only {\it dangerous} reeds are those containing a zero-momentum line,
 thus we  need to provide a suitable renormalisation
only for such reeds.
We say that a line $\ell$ is \emph{regular} if $\mu_\ell\ne0$, otherwise it is \emph{singular}.
Since $V_0=0$, the order coincides with the number of nodes and
moreover two singular lines cannot be connected;
however it is possible that there is only one regular line between two singular lines: in such
a case we shall say that the two singular lines are \emph{consecutive}.
Given two consecutive singular lines $\ell_1,\ell_2$ we call \emph{link} the line $\ell$ between $\ell_1$ and
$\ell_2$, and we say that the link has size $n\ge1$ if $|\mu_\ell|=n$. Note that by the conservation law
\eqref{conservareed} a link can have only momentum $\mu_\ell \neq0$.

For all $n\in\ZZZ\setminus\{0\}$ set
\begin{equation}\label{jav}
\ol{j}_n(\x):= j_n(\x)+ j_{-n}(\x)\,.
\end{equation}

Given a reed $\rho$ with a link $\ell_0$  of size $n$, let $\rho'$ be the reed obtained by exchanging the modes
of the nodes $\ell_0$ enters and exits respectively, so that
$\mu_{\ell_0}$ is replaced with $-\mu_{\ell_0}$;
note that this replacement effects only the propagator of the link, while the product of all the node factors and the
other propagators is the same. In other words we have
\be\label{sommetta}
\begin{aligned}
\Val(\rho)+\Val(\rho')&=\Big(\prod_{v\in N(\rho)}\calF_v\Big)\Big(\prod_{\ell\in L(\rho)\setminus\{\ell_0\}} \calG_{\ell}\Big)
\big(j_{\mu_{\ell_0}}(\x)+j_{-\mu_{\ell_0}}(\x)\big)\\
&=\Big(\prod_{v\in N(\rho)}\calF_v\Big)\Big(\prod_{\ell\in L(\rho)\setminus\{\ell_0\}} \calG_{\ell}\Big) \ol{j}_n(\x)\,.
\end{aligned}
\ee
If we perform this procedure for all reeds, we can rewrite \eqref{ovvio} as
\begin{equation}\label{menovvio}
\psi_{N,\mu}= \sum_{\rho\in \Theta^*_{N,\mu}} \Val^*(\rho)\,,
\end{equation}
where $\Theta^*_{N,\mu}$ is the set of reeds with order $N$ and total momentum $\mu$ in which
the links have only positive momentum and their propagators are replaced by $\ol{j}_n(\x)$ in \eqref{jav}.
We are now ready to perform the renormalisation procedure.

Denote by $\Theta^{\RR}_{N,\mu}$ the set of reeds in which no link appears,
set $j_0^{\RR_0}(\x):=j_0(\x)$ and for $n\ge 1$ define recursively
\begin{equation}\label{reno}
j_0^{\RR_n}(\x):= \frac{j_0^{\RR_{n-1}}(\x)}{1+\g^2|V_{n}|^2\ol{j}_n(\x) j_0^{\RR_{n-1}}(\x)}\,.
\end{equation}

\begin{rmk}\label{guardatelo}
Note that
$$
j_0^{\RR_n}(\x)=\frac{j_0(\x)}{1+\g^2(\sum_{p=1}^{n}|V_p|^2\ol{j}_p(\x))j_0}
$$
so that we can define
\begin{equation}\label{staqua}
j_0^\RR(\x)=\lim_{n\to\io}j_0^{\RR_n}(\x)=\frac{j_0(\x)}{1+\g^2(\sum_{p\ge1}|V_p|^2\ol{j}_p(\x))j_0(\x)}\,.
\end{equation}
Since $\g^2(\sum_{p\ge1}|V_p|^2\ol{j}_p(\x))$ is purely imaginary while
$j_0(\x)$ is real, the limit in \eqref{staqua} is well defined and moreover using \eqref{jk} one has
\begin{equation}\label{boundren}
|j_0^\RR(\x)|\le (\g^2\sum_{p\ge1}|V_p|^2|j_p(\x)|)^{-1}\le \frac{\sqrt{2\ol{\e}}}{\g^{2} \|V\|^{2}_{L^2}}\,.
\end{equation}

\end{rmk}

For any $\r\in\Theta^\RR_{N,\mu}$ let us define the renormalised value of $\rho$ as
\be\label{renval}
\Val^\RR(\rho)=\prod_{i=1}^N \calF_{v_i}\calG^{\RR}_{\ell_i}\,,
\ee
where
\be\label{renprop}
\calG^{\RR}_{\ell_i}=\left\{
\begin{aligned}
&j_{\mu_{\ell_i}}(\x),\qquad \mu_{\ell_i}\ne0\,,\\
&j_0^{\RR}(\x), \qquad \mu_{\ell_i}=0\,.
\end{aligned}
\right.
\ee

\begin{rmk}\label{contala}
Note that by construction in a renormalised reed with $N$ nodes one can have at most
$N/3$ lines carrying zero-momentum: indeed between two zero-momentum line there must
be at least two nonzero-momentum lines.
\end{rmk}

As in the non-degenerate case, setting
\be\label{coeffrin}
\psi_\mu^\RR(\x;\g):=\sum_{N\ge1}(-\ii\g)^N
\sum_{\r\in\Theta^\RR_{N,\mu}}\Val^\RR(\rho)\psi_0(\x)\,,
\ee
it is, again formally
$$
\psi_\mu(\x;\g)=\psi_\mu^\RR(\x;\g)\,.
$$
Moreover, using Remark \ref{contala} we have
\begin{equation}\label{mestavoascorda1}
\begin{aligned}
|\Val^\RR(\r)|&=\Big(\prod_{v\in N(\rho)} |\calF_v|\Big) \Big(\prod_{\ell\in L(\rho)}| \calG_\ell|\Big)
\le \Big( C_0e^{-\s\sum_{v\in N(\rho)}|n_v|}\Big) \Big(\prod_{\ell\in L(\rho)}| \calG_\ell|\Big)\\
&\le 
e^{-\s\sum_{v\in N(\rho)}|n_v|}\left(\frac{\sqrt{2\ol{\e}}}{\g^{2} \|V\|^{2}_{L^2}}\right)^{N/3}\left(\frac{1}{\sqrt{2\ol{\e}}}\right)^{2N/3}\\
&\le
\left(\frac{C_0}{ C_2(2\ol{\e})^{1/6}}\right)^{N}\g^{-2 N/3}e^{-\s|\mu|}\,,
\end{aligned}
\end{equation}
which implies that there exists a constant $C_2>0$ such that
\be\label{stimatotale1}
|\psi_\mu^\RR(\x;\g)|\le\sum_{N\ge1}\g^{N/3} \left(\frac{C_0}{ C_2(2\ol{\e})^{1/6}}\right)^{N}e^{-\s|\mu|/2}
\ee
and we see that the series above converges for
\begin{equation}\label{gammanon1}
\g\le (C_0^{-1}{ C_2}(2\ol{\e})^{6})^{3}\,.
\end{equation}

This means that the function
\begin{equation}\label{asy1}
\psi^\RR(\f;\x,\g):=
\sum_{\mu\in\ZZZ}e^{\ii\mu\f}\psi_\mu^\RR(\x;\g)\,,
\end{equation}
is well defined, uniformly bounded for $\x\in[-1,1]$ and $\g$ satisfying \eqref{gammanon1}
and analytic for $\f\in\TTT$.

\subsection{The asymptotic solution}
\label{fine}

From the previous discussion we obtained the convergence of the series
\be\label{serierin}
\psi^\RR(\f;\x,\g):=\sum_{\mu\in\ZZZ}e^{\ii\mu\f}\psi_\mu^\RR(\x;\g)
\ee
in both the non-degenerate case and the perturbative regime.
We complete our proof of Proposition \ref{asintoticsol} showing that the function $\psi^\RR(\f;\x,\g)$ is in fact a solution of the equation \eqref{dua}. This is essentially a straightforward computation.

First we notice that the difference
between \eqref{coeffrin1} and \eqref{coeffrin} is only in the definition of the set $\Theta^\RR_{N,\mu}$ and the 
explicit expression of the renormalised propagators. Indeed using either \eqref{coeffrin1} or \eqref{coeffrin} (depending on the case) into \eqref{dua} we get
\be\label{daje}
\begin{aligned}
\Big[\int_{-\infty}^\f 
&\dd \f' J_0\left(\a(\f-\f')\right)e^{\ii\a \x (\f-\f')}(V_0+\ii\g V(\f'))\psi^\RR(\f';\x,\g)
\Big]_\mu\\
&=\sum_{N\ge0}(-\ii\g)^N
\sum_{\r\in\Theta^\RR_{N,\mu}}\Val^\RR(\rho)\psi_0(\x)\\
&=\hp_0(\x)-\int_{-\infty}^\f \!\!\!\!\!\!
\dd \f' J_0\left(\a(\f-\f')\right)e^{\ii\a \x (\f-\f')}(V_0+\ii\g V(\f'))\times\\
&\qquad\qquad\qquad\times
\sum_{N\ge1}(-\ii\g)^N
\!\!\!\!
\sum_{\r\in\Theta^\RR_{N,\mu}}\!\!\!\!
\Val^\RR(\rho)\hp_0(\x)\\
&=\Big(1+
\sum_{N\ge1}(-\ii\g)^{N}\int_{-\infty}^\f \!\!\!\!\!\!
\dd \f' J_0\left(\a(\f-\f')\right)e^{\ii\a \x (\f-\f')}(V_0+\ii\g V(\f'))\times\\
&\qquad\qquad\qquad\times
\sum_{\r\in\Theta^\RR_{N,\mu}}\!\!\!\!
\Val^\RR(\rho)
\Big)\hp_0(\x)\,,
\end{aligned}
\ee
which means that the coefficients must satisfy
\be\label{perico}
\begin{aligned}
\sum_{N\ge1}&(-\ii\g)^N
\sum_{\r\in\Theta^\RR_{N,\mu}}\Val^\RR(\rho)=\\
&=\sum_{N\ge1}(-\ii\g)^{N}\int_{-\infty}^\f \!\!\!\!\!\!
\dd \f' J_0\left(\a(\f-\f')\right)e^{\ii\a \x (\f-\f')}(V_0+\ii\g V(\f'))\sum_{\r\in\Theta^\RR_{N-1,\mu}}\!\!\!\!
\Val^\RR(\rho)\,.
\end{aligned}
\ee

On the other hand the r.h.s. in \eqref{perico} equals 
$$
\sum_{N\ge1}(-\ii\g)^{N}
\sum_{\r\in\ol{\Theta}^\RR_{N,\mu}}
\Val^\RR(\rho)
$$
where $\ol{\Theta}^\RR_{N,\mu}$ is the set of reeds whose root line may exits a zero-mode node (in the
non-degenerate case) or the line immediately preceding the root line may be a link (in the perturbative regime).

We study separately the two cases.

\subsubsection{The non-degenerate case.}

In this case the root line of a reed has to be renormalised only if it exits a zero-mode node.

Concerning the case $|\mu|<2\al+\ol{\e}$, equation \eqref{perico} reads
\be\label{sforzo1}
\psi^\RR_\mu(\x;\g)= \sum_{N\ge1}(-\ii\g)^{N}
\sum_{\r\in\ol{\Theta}^\RR_{N,\mu}}
\Val^\RR(\rho)\,.
\ee

Let us split
\be\label{split}
\ol{\Theta}^\RR_{N,\mu}=\widetilde{\Theta}^\RR_{N,\mu}\cup \hat{\Theta}^\RR_{N,\mu}
\ee
where
$\hat{\Theta}^\RR_{N,\mu}$ is the set of reeds such that the root line exits a node carrying zero
mode and $\widetilde{\Theta}^\RR_{N,\mu}$ are all the others reeds in $\ol{\Theta}^\RR_{N,\mu}$:
note that if $\rho \in \widetilde{\Theta}^\RR_{N,\mu}$ then all the lines in $\rho$ are renormalised.

Therefore we have
\be\label{pezzouno1}
\sum_{N\ge1}(-\ii\g)^N\sum_{\rho\in\widetilde{\Theta}^\RR_{N,\mu}} \Val^\RR(\rho) = 
-\ii\g j_\mu(\x) \sum
 V_{\mu_1} \psi_{\mu_2}^\RR(\x;\g)
\ee
and
\be\label{pezzodue}
\sum_{N\ge1}(-\ii\g)^N\sum_{\rho\in\hat{\Theta}^\RR_{N,0}} \Val^\RR(\rho) =
-\ii\g j_0(\x) \sum_{\mu\in\ZZZ\setminus\{0\}}(\ii\g)^2V_\mu j_{-\mu}(\x) V_{-\mu} \psi_{\mu}^\RR(\x;\g)\,,
\ee
so that summing together \eqref{pezzouno} and \eqref{pezzodue} we obtain $\psi^\RR_0(\x;\g)$.


\subsubsection{The moderately resonant regime.}

In this case the root line of a reed has to be renormalised only if it carries zero momentum,
thus for $\mu\ne0$ we see immediately that \eqref{perico} holds.

Concerning the case $\mu=0$, equation \eqref{perico} reads
\be\label{sforzo}
\psi^\RR_0(\x;\g)= \sum_{N\ge1}(-\ii\g)^{N}
\sum_{\r\in\ol{\Theta}^\RR_{N,0}}
\Val^\RR(\rho)\,.
\ee

Let us split
\be\label{split}
\ol{\Theta}^\RR_{N,0}=\widetilde{\Theta}^\RR_{N,0}\cup \hat{\Theta}^\RR_{N,0}\,,
\ee
where $\widetilde{\Theta}^\RR_{N,0}$ is the set of reeds in which no link appears while
$\hat{\Theta}^\RR_{N,0}$ are the reeds such that the line immediately preceding the root line
is a link. Therefore we have
\be\label{pezzouno}
\sum_{N\ge1}(-\ii\g)^N\sum_{\rho\in\widetilde{\Theta}^\RR_{N,0}} \Val^\RR(\rho) = 
-\ii\g j_0(\x) \sum_{\mu\in \ZZZ\setminus\{0\}} V_{-\mu} \psi_\mu^\RR(\x;\g)
\ee
and
\be\label{pezzodue}
\sum_{N\ge1}(-\ii\g)^N\sum_{\rho\in\hat{\Theta}^\RR_{N,0}} \Val^\RR(\rho) =
-\ii\g j_0(\x) \sum_{\mu\in\ZZZ\setminus\{0\}}(\ii\g)^2V_\mu j_{-\mu}(\x) V_{-\mu} \psi_{\mu}^\RR(\x;\g)\,,
\ee
so that summing together \eqref{pezzouno} and \eqref{pezzodue} we obtain $\psi^\RR_0(\x;\g)$.

This concludes the proof of Proposition \ref{asintoticsol}.

\begin{rmk}\label{rmk:mean}
As anticipated in the introduction, the solution in the moderately resonant regime exhibits a non-zero mode even though the forcing has vanishing 
average, as it can be directly verified by (\ref{sforzo}). Since $\ol{\Theta}^\RR_{N,0}$ is non empty as $N\geq2$, the 
average of the solution is $\OOO{\frac {h^2} {\w^2}}$\,. 
\end{rmk}

\section{Finite time analysis}\label{sect:finite-time}

To conclude the proof of Theorem \ref{main-as}, we need to show that the solution for finite $t_0$
is asymptotic as $t_0\to \io$ to $\psi_\infty(\x,t)$.
More precisely, the goal of this section is to prove the following proposition.
\begin{prop}\label{prop:asintotica}
For any $t\in\RRR$, $\xi\in[-1,1]$ one has
\be\label{eq:delta}
\lim_{t_0\to-\infty}\sqrt{t-t_0}|\psi_\infty(\x,t)-\psi_{t_0}(\x,t)|< \infty.
\ee
\end{prop}

We need some preliminary considerations. It is convenient to define
\be\label{eq:def-delta}
\Delta_{t_0}(t,\x):=\psi_\infty(\x,t)-\psi_{t_0}(\x,t)\,.
\ee
Note that as a difference of two analytic functions, $\D$ is analytic in $h$. Thus there is a sequence of functions $\{\D^k\}_{k\in \N_0}$ such that
\be\label{eq:esp.delta}
\D=\sum_{j\geq0} (ih)^j\D^j\,.
\ee

Combining \eqref{duat0} and \eqref{duat} we get
$$
(\uno+\ii W_\infty)\psi_\infty(\x,t)=(\uno+\ii W_{t_0})\psi_{t_0}(\x,t)\,,
$$
thus $\D_{t_0}(t,\x)$ satisfies
\be\label{eq:eq.delta}
(\uno+\ii W_{t_0})\D_{t_0}(\x,t)=-\ii q^{[0]}(\x,t,t_0,\w)\,,
\ee
with
\be\label{eq:q}
q^{[0]}(\x,t,t_0,\w):=\int_{-\infty}^{t_0}\dd \tau J_0(g(t-\tau))e^{i\x(t-\tau)}(V_0+ih V(\tau))\psi_\infty(\x,\tau)\,.
\ee
Let us now look at $q^{[0]}(t,t_0,\xi,\w)$. Expanding in Fourier series we get
\bea
q^{[0]}(t,t_0,\xi,\w)&=&V_0\sum_{k\in\Z}\hat\psi_ke^{\ii\w kt}\int_{t-t_0}^{\infty}\dd \tau J_0(g\tau) e^{i(\xi-\w k)\t}\nn\\
&+&h\sum_{\mu,k\in\ZZZ}\hat\psi_\mu \hat V_k e^{\ii\w(\mu+k)t} \int_{t-t_0}^{\infty}\dd \tau J_0(g\tau) e^{i(\xi-\w (\mu+k))\t}\nn\\
&=&\sum_{n\in\ZZZ}\left(V_0\hat\psi_n+h(\hat\psi^{\infty}\ast \hat V)_n\right)e^{\ii\w nt} \int_{t-t_0}^{\infty}\dd \tau J_0(g\tau) e^{i(\xi-\w n)\t}\,.\nn
\eea
Therefore the function $q^{[0]}(t,t_0,\xi,\w)$ can be written as
\be\label{eq:q-dec}
q^{[0]}(t,t_0,\xi,\w)= \sum_{n\in\ZZZ} e^{\ii \w n t}q^{[0]}_n(t,t_0,\xi,\w)\,,
\ee
with
\be\label{eq:qn}
q^{[0]}_n(t,t_0,\xi,\w):=\left(V_0\psi_n+h(\psi_{\infty}\ast \hat V)_n\right)\int_{t-t_0}^{\infty}\dd \tau J_0(g\tau) e^{i(\xi-\w n)\t}\,.
\ee
We first study the decay properties of $q^{[0]}(\x,t,t_0,\w)$. 

\begin{lemma}\label{lemma:dec-q0}
There exists a bounded function $\tilde r(t,t_0,\xi,\w)$ such that 
\be\label{eq:dec-q0}
-\ii q_n^{[0]}(t,t_0,\xi,\w)=(\psi_{\infty,n}-\d_{n0}\psi_0)\left[\sum_{\s=\pm1}\frac{e^{\ii(\x-\w n+\s)(t-t_0)}\tilde r(t,t_0,\xi,\w)}{\sqrt{t-t_0}}+\OO{t-t_0}\right]\,.
\ee
\end{lemma}

\begin{proof}
By (\ref{eq:qn}) we see that apparently the integral on the r.h.s. may have some divergences, but these are in fact suppressed by $\psi_\infty$. Indeed we observe
$$
\d_{n0}=\psi_{\infty,n}+\left(\ii V_0\psi_n+ih(\psi_{\infty}\ast\hat V)_n\right)\int_{0}^{\infty}\dd \tau J_0(g\tau) e^{i(\xi-\w n)\t}\,.
$$
Both the terms of this equality are finite uniformly in $\xi$. In fact  the r.h.s. is regular for all $\xi\in[-g,g]$ except a finite set of values $I_n(\w)$. Thus we can write for any $\xi\in[-1,1]\setminus I_n(\a)$
$$
iV_0\psi_n+ih(\psi\ast \hat V)_n=-(\psi_{\infty,n}-\d_{n0})\left( \int_0^\infty \dd \t J_0(g\t)e^{i(\xi-\w n)\t} \right)^{-1}\,.
$$
Now we plug the last equality into (\ref{eq:qn}) obtaining
\be\label{eq:q^[0]}
-\ii q^{[0]}_{n}(t,t_0,\xi,\w)=(\psi^{\infty}_n-\d_{n,0})\frac{\int_{t-t_0}^{\infty}\dd \tau J_0(g\tau) e^{\ii(\xi-\w n) \t}}{\int_0^\infty \dd \t J_0(g\t)e^{\ii(\xi-\w n)\t}}\,.
\ee
This is a well defined expression for $\xi\in(-g,g)$. Moreover, thanks to (\ref{eq:TFJ0}) and (\ref{eq:LEMMA-J0-A-INFTY}), we can extend by continuity to $\xi\in I_n(\a)$, obtaining (\ref{eq:dec-q0}).
\end{proof}

\begin{rmk}
Since $\psi_\infty$ is analytic, all the $q^{[0]}_{n}$ are exponentially small in $n$. Therefore we readily get the overall decay
\be\nn
q^{[0]}(t,t_0,\xi,\w)=\OO{\sqrt{t-t_0}}\,. 
\ee
\end{rmk}

The non-degenerate case can be discussed as follows. Set
$$
\tilde W_{t_0} f:=\int_{t_0}^t\dd t' J_0(g(t-t'))e^{ig\x(t-t')}V(\w t')f(\x, t')\,,
$$
so that, using \eqref{Wbar}, we get
 $W_{t_0}=iV_0\bar W_{t_0}+ih\tilde W_{t_0}$. By  (\ref{eq:q^[0]}), also $q^{[0]}$ is analytic in $h$, with
\be\label{eq:coeff-q0}
q^{[0]}=\sum_{j\geq0} (ih)^jq^{[0],j}\,. 
\ee
The last equality must be understood as a definition of the $q^{[0],j}$, with each of these functions decaying
 for large $t_0$ according to Lemma \ref{lemma:dec-q0}.

Plugging (\ref{eq:esp.delta}) and (\ref{eq:coeff-q0}) into (\ref{eq:def-delta}) and using the resolvent identity (which can be directly verified to hold for $h$ small enough) we get the following sequence of equations
\be\label{eq:eqdeltaV0}
\left\{
\begin{array}{rcl}
(\II+iV_0\bar W_{t_0})\d^0&=&q^{[0],0}\,,\\
(\II+iV_0\bar W_{t_0})\d^{k}&=&q^{[0],k}-i\tilde W_{t_0}\D^{k-1}\,,\quad k\geq1\,.
\end{array}\right. 
\ee

The study of these equations can be carried out in analogy to what was done in Appendix \ref{B} (from (\ref{eq:h=0-cnv}) onward), with the aid of the forthcoming analysis of the moderately resonant case to control the operator $\tilde W_{t_0}$. However a detailed exposition would bring fairly long computations bearing nothing new with respect to Section 2 and we prefer to omit it. The conclusion is that  $\D_k=\OO{\sqrt{t-t_0}}$ for all $k\ge1$, therefore also $\D$ is so. This proves Proposition \ref{prop:asintotica} in the non-degenerate case. 

Let us now discuss in detail the case $V_0=0$. Since $W_{t_0}$ is compact for any finite $t_0$, we can invert (\ref{eq:eq.delta}) by Neumann series: 
\be\label{eq:res-exp-delta}
\D_{t_0}(\x,t)=(\uno+iW_{t_0})^{-1}(-ih q^{[0]}(\x,t,t_0,\w))=-\sum_{k\geq0}(ih)^{k+1} W^k_{t_0}q^{[0]}(\x,t,t_0,\w)\,.
\ee
We set for brevity
$$
q^{[k]}(\x,t,t_0,\w):=W^k_{t_0}q^{[0]}(\x,t,t_0,\w)\,.
$$
The same decomposition of $q^{[0]}(\x,t,t_0,\w)$ is extended to any $q^{[k]}(\x,t,t_0,\w)$ by a similar calculation.
\begin{lemma}\label{lemma:qkn}
For any $k\geq0$. Then 
$$
q^{[k]}(\x,t,t_0,\w)=\sum_{n\in\ZZZ}e^{\ii n\w t} q^{[k]}_n(\x,t,t_0,\w)\,,
$$
with the functions $q^{[k]}_n(\x,t,t_0,\w)$ for $k\geq1$ recursively defined as
\be\nn
q_n^{[k]}(\x,t,t_0,\w):=\int_0^{t-t_0} dt'J_0(gt')e^{\ii(\xi-\w n)t'}(\hat V\ast q^{[k-1]})_n(t-t',t_0,\xi,\w)\,.
\ee
\end{lemma}
\begin{proof}
We apply $W_{t_0}$ to $q^{[0]}$ and get
\bea
q^{[1]}(\x,t,t_0,\w)&=&\int_0^{t-t_0}dt' J_0(gt')e^{\ii\xi t'}V(t-t')q^{[0]}(\x,t-t',t_0,\w)  \nn\\
&=&\sum_{n,m\in\Z}e^{\ii\w(n+m) t}\int_0^{t-t_0}dt' J_0(gt')e^{\ii(\xi-\w (n+m))t'}\hat V_mq_n^{[0]}(\x,t-t',t_0,\w)\nn\\
&=&\sum_{n\in\Z}e^{\ii\w nt}\int_0^{t-t_0}dt' J_0(gt')e^{\ii(\xi-\w n)t'}(\hat V\ast q^{[0]})_n(\x,t-t',t_0,\w)\,.\label{eq:q1}
\eea
In the same way we get the proof for all $k\geq1$.
\end{proof}

\begin{lemma}\label{lemma:indWt01}
One has
$$
q_n^{[k+1]}(t,t_0,\xi,\w)=\sum_{\substack{m\in\Z \\ \s,\s'=\pm1}}(\hat V\ast^{k}\psi_\infty)_m\hat V_{n-m}e^{\ii(\xi-\w n+\s')(t-t_0)}f_{\s'}(\w(n-m)t'+\s,t-t_0)+\OO{t-t_0}\,.
$$
\end{lemma}
\begin{proof}
We perform the proof by induction on $k$. 
For $k=1$ we use (\ref{eq:dec-q0}) into (\ref{eq:q1}) and we get
\bea
q_n^{[1]}(t,t_0,\xi,\w)&=&\int_0^{t-t_0}dt' J_0(gt')e^{\ii(\xi-\w n)t'}(\hat V\ast q^{[0]})_n(\x,t-t',t_0,\w)\nn\\
&=&\sum_{\substack{m\in\Z \\ \s=\pm1}}(\hat V\ast\psi_\infty)_m\hat V_{n-m}e^{\ii(\xi-\w m+\s)(t-t_0)}\int_0^{t-t_0}dt' \frac{J_0(gt')}{\sqrt{t-t'-t_0}}e^{-\ii (\w(n-m)t'+\s) t'}\nn\\
&+&\OO{t-t_0}\nn\\
&=&\sum_{\substack{m\in\Z \\ \s,\s'=\pm1}}(\hat V\ast\psi_\infty)_m\hat V_{n-m}e^{\ii(\xi-\w n+\s')(t-t_0)}f_{\s'}(\w(n-m)t'+\s,t-t_0)\label{eq:q1-as}\\
&+&\OO{t-t_0}\nn\,.
\eea

Here the oscillatory integral is estimated by Lemma \ref{cor:coda-TFJ0}. 
Note that the phase factor in (\ref{eq:q1-as}) is the same as the one 
appearing in (\ref{eq:dec-q0}). Therefore an analogous argument for $k\geq2$ (in virtue of Lemma \ref{lemma:qkn}) proves that if
$$
q_n^{[k]}(t,t_0,\xi,\w)=\sum_{\substack{m\in\Z \\ \s,\s'=\pm1}}(\hat V\ast^{k-1}\psi_\infty)_m\hat V_{n-m}e^{\ii(\xi-\w n+\s')(t-t_0)}f_{\s'}(\w(n-m)t'+\s,t-t_0)+\OO{t-t_0}
$$
then
$$
q_n^{[k+1]}(t,t_0,\xi,\w)=\sum_{\substack{m\in\Z \\ \s,\s'=\pm1}}(\hat V\ast^{k}\psi_\infty)_m\hat V_{n-m}e^{\ii(\xi-\w n+\s')(t-t_0)}f_{\s'}(\w(n-m)t'+\s,t-t_0)+\OO{t-t_0}\,.
$$
Therefore the assertion follows.
\end{proof}

\begin{proof}[Proof of Proposition \ref{prop:asintotica}]
We want to prove that $|\D|=\OO{\sqrt{t-t_0}}$. We will use the decomposition (\ref{eq:res-exp-delta}) which gives
$$
\D=ih \sum_{k\geq0}(ih)^kq^{[k]}.
$$
The boundedness of each $q^{[k]}(t,t_0,\xi,\w)$ is given by the analyticity property of $V$ and $\psi_\infty$. Moreover a standard convolution estimate gives
$$
|q^{[k]}(t,t_0,\xi,\w)|\lesssim C_0^k\,,
$$
uniformly in $\w,\xi$. This yields the convergence of the series (\ref{eq:eq.delta}) at fixed $t,t_0$ for $h$ sufficiently small as in Proposition \ref{asintoticsol} (see (\ref{gammanon1})). Finally, the terms in which $\w(n-m)t'+\s$ is equal to $\pm1$ are vanishing since $\hat V_0=0$. So, having in mind Lemma \ref{cor:coda-TFJ0} we get the desired behaviour  $\OO{\sqrt{t-t_0}}$ of each $q^{[k]}$ and hence of $\D$, uniformly in $\w,\xi$. 
\end{proof}

\appendix

\section{Some oscillatory integrals}\label{App:Bessel}

In this appendix we present explicit computations of some one-dimensional oscillatory integrals which are used in the paper.
Recall the definition of the Fourier and Hilbert transforms $\matF$ and $\matH$ in \eqref{fouhil}.
To lighten the notation, we convey henceforth that every time $\frac1t$ appears, it will be understood in the sense of principal value.

We start by recalling the following elementary identity holding for any $a>0$:
\be\label{eq:heav-a}
\int_0^a \dd t e^{i\t t}=\frac{\sin a\t}{\t}+i\frac{1-\cos a\t}{\t}\,.
\ee

Set for brevity
$$
\sinc \t:=\frac{\sin \t}{\t}\,,\quad \cosc \t:=\frac{1-\cos \t}{\t}\,,
$$
and recall that $t^{-1}-\cosc(at)-i\sinc(at)=\frac{e^{-iat}}{t}$ tends to $\mp i\pi\d(t)$ as $a\to\pm\infty$ in the sense of distributions.
Moreover as $a\to\infty$ (\ref{eq:heav-a}) gives the Fourier transform (modulo a factor $\sqrt{2\pi}$) $H^>(\t)$ of the Heaviside function $\chi(t>0)$, that is
\be\label{eq:Heaviside}
H^>(\t):=\int_0^\infty \frac{\dd t}{2\pi} e^{i\t t}=\frac12\d(\t)+\frac{i}{2\pi t}\,,
\ee
where the above equality must be obviously understood in a weak sense. 
Let us denote 
\begin{equation}\label{notaH}
H^<(\t):=\matF[1_{\{t<0\}}](\t),\qquad H^\gtrless_{a}(\t):=\matF[1_{\{t\gtrless0\}}](\t)
\end{equation}
 and note that $H^<(\t)=H^>(-\t)$ and $H^\gtrless_{a}(\t)=e^{i\t a}H^\gtrless(\t)$. A worth feature of the Hilbert transform is that it acts as a Fourier multiplier:\be\label{eq:H-moltiplicatore}
\matH[e^{i\t t}]=\sqrt{\frac2\pi}e^{i\t t}\left(\matF\left[\frac{1}{(\cdot)}\right](\t)\right)^*=-ie^{i\t t}\sign\t\,,
\ee
since we choose
$$
\matF\left[\frac{1}{(\cdot)}\right](\t)=i\sqrt{\frac\pi2}\sign\t\,.
$$
as normalizing factor.
Using the notation of Section 2 we set $f_a:=fe^{-iat}$. Our first result is the following.
\begin{lemma}\label{lemma:H[d_a/t]}
Let $f\in L^1_{\loc}(\R)$, $a\in\R$. We have
\be\label{eq:important0-a-Hilbert}
\matH\left[\frac{e^{-iat}}{t}f(t)\right]=-\frac{i}{\sqrt{2\pi t}}(\matF[f]\ast\sign)(a)-\pi e^{iat}\d(t)f(t)+\left(\frac{e^{2iat}}{ t}\right)\matH[f_a]\,.
\ee
\end{lemma}
\begin{proof}
First we note that
$$
\matF\left[\frac{e^{-ia(\cdot)}}{(\cdot)}\right](\t)=i\sqrt{\frac\pi2}\sign(\t+a)\,.
$$
Then
\bea
-\matH\left[\frac{e^{-iat}}{t}f(t)\right]&=&-\sqrt{\frac\pi2}\int_\R \frac{d\t}{\sqrt{2\pi}} \int_\R \frac{d\t'}{\sqrt{2\pi}} i\sign(\t+a)\matF[f](\t')\matH\left[e^{-i(\t+\t')t}\right]\nn\\
&=&\sqrt{\frac\pi2}\int_\R \frac{d\t}{\sqrt{2\pi}} \int_\R \frac{d\t'}{\sqrt{2\pi}} \sign(\t+a)\sign(\t+\t')\matF[f](\t')e^{-i(\t+\t')t}\nn\,.
\eea
We split the integral into several pieces with definite sign
\bea
-\sqrt{\frac2\pi}\matH\left[\frac{e^{-iat}}{t}f(t)\right]&=&\int_{a}^{+\infty} \frac{d\t'}{\sqrt{2\pi}}\matF[f](\t') \int_{-a}^{\infty} \frac{d\t}{\sqrt{2\pi}} e^{-i(\t+\t')t}+\int_{a}^{+\infty} \frac{d\t'}{\sqrt{2\pi}}\matF[f](\t')\int^{-\t'}_{-\infty} \frac{d\t}{\sqrt{2\pi}} e^{-i(\t+\t')t}\nn\\
&-&\int^{\infty}_{a} \frac{d\t'}{\sqrt{2\pi}}\matF[f](\t')\int_{-\t'}^{-a}e^{-i(\t+\t')t}+\int^{a}_{-\infty} \frac{d\t'}{\sqrt{2\pi}}\matF[f](\t') \int^{-a}_{-\infty} d\t e^{-i(\t+\t')t}\nn\\
&+&\int^{a}_{-\infty} \frac{d\t'}{\sqrt{2\pi}}\matF[f](\t') \int_{-\t'}^{\infty} d\t e^{-i(\t+\t')t}-\int_{-\infty}^{a} \frac{d\t'}{\sqrt{2\pi}}\matF[f](\t')\int^{-\t'}_{-a}e^{-i(\t+\t')t}\nn\\
&=&\int_{a}^{+\infty} d\t'\matF[f](\t')e^{-i\t't}\left(2e^{iat}H^<+e^{i\t' t}(H^>-H_<)\right)\nn\\
&+&\int^{a}_{-\infty} d\t'\matF[f](\t') e^{-i\t't}\left(2e^{iat}H^>-e^{i\t' t}(H^>-H_<)\right)\nn\,.
\eea
Using (\ref{eq:Heaviside}) we have
\bea
&&2e^{iat}H^<\int_{a}^{+\infty} d\t'\matF[f](\t')e^{-i\t't}+2e^{iat}H^>\int^{a}_{-\infty} d\t'\matF[f](\t') e^{-i\t't}\nn\\
&=&\sqrt{2\pi}e^{iat}\d(t)f(t)-\frac{1}{\pi}\pv\left(\frac{e^{iat}}{\t}\,, \cdot\right)\int_\R d\t'\matF[f](\t')e^{-i\t't} i\sign(\t'-a)\nn\\
&=&\sqrt{2\pi}e^{iat}\d(t)f(t)-\sqrt{\frac2\pi}\frac{1}{\pi}\pv\left(\frac{e^{iat}}{t}\,, \cdot\right)\left(f\ast\frac{e^{ia(\cdot)}}{(\cdot)}\right)(t)\nn\,.
\eea
Furthermore
\bea
&&\int_{a}^{+\infty} d\t'\matF[f](\t')(H^>-H_<)-\int_{-\infty}^a d\t'\matF[f](\t')(H^>-H_<)\nn\\
&=&\pv\left(\frac{i}{\pi t}\,, \cdot\right)(\matF[f]\ast\sign)(a)\nn\,.
\eea
Recollecting all the terms we obtain (\ref{eq:important0-a-Hilbert}). 
\end{proof}

\begin{rmk}
Formula (\ref{eq:important0-a-Hilbert}) has two interesting upshots:
\begin{itemize}
\item[i)] as $a\to\pm\infty$ the second and the third summand vanish and we recover
\be\label{eq:H-delta}
\matH[\d f](t)=\pm\pv\frac{f(0)}{\pi t}\,,
\ee
which can be obtained by the very definition of Hilbert transform. \\

\item[ii)] as $a=0$ we get
\be\label{eq:H[f/t]}
\matH\left[\frac{f}{(\cdot)}\right](t)=\pv\frac{\matH[f]-\matH[f](0)}{t}-\pi\d(t)f(t)\,.
\ee
If $f(t)$ is constant (say $f=1$) then 
\be\label{eq:H[1/t]}
\matH\left[\frac{1}{(\cdot)}\right](t)=-\pi\d(t)\,.
\ee

Finally from (\ref{eq:H[f/t]}) using (\ref{eq:H[1/t]}) we obtain
\be\label{eq:H-ultima-utile}
\matH\left[\frac{f(0)-f}{(\cdot)}\right](t)=-\pi(f(0)-f(t))\d(t)+\pv\frac{\matH[f](0)-\matH[f](t)}{t}\,. 
\ee
As $t\to0$ the last equation gives the usual derivation rule of the Hilbert transform $H[f]'=H[f']$ (see for instance \cite{king}). 
\end{itemize}
\end{rmk}

\vspace{1cm}

Now we pass to examine some integrals involving the Bessel function $J_0$, defined through
\be\label{eq:reprJ0}
J_0(x)=\int_{-\pi}^{\pi}\frac{d\theta}{\pi}e^{i\cos\theta x}\,.
\ee

We have
\begin{lemma}\label{lemma:TFJ0}
Let $a,\t\in\RRR$. It holds
\begin{align}
\int_0^\infty \dd tJ_0(t)e^{i\t t}&=\frac{\chi(|\t|\leq1)+i\chi(|\t|\geq1)}{\sqrt{|1-\t^2|}}\label{eq:TFJ0}\,,\\
\int_0^a \dd tJ_0(t)e^{i\t t}&=\int_{\t-1}^{\t+1} \frac{\dd x}{\pi}\frac{a\sinc(ax)}{\sqrt{1-(x-\t)^2}}+i\int_{\t-1}^{\t+1} \dd x\frac{a\cosc(ax)}{\sqrt{1-(x-\t)^2}}\,.\label{eq:TFJ0-a}
\end{align}
\end{lemma}

\begin{rmk}
Note that the r.h.s. of (\ref{eq:TFJ0}) is not an integrable function and the above formula cannot be interpreted as a Fourier transform. 
Indeed the Fourier transform takes into account only the real part of (\ref{eq:TFJ0}), namely
$$
J_0(t)=\int_{-1}^1\frac{e^{it\t}\dd \t}{\sqrt{1-\t^2}}\,,\quad \matF[J_0]=\sqrt{\frac2\pi}\frac{\chi(|\t|<1)}{\sqrt{1-\t^2}}\,. 
$$
\end{rmk}

\begin{proof}
Let us compute
\bea
\int_0^\infty \dd tJ_0(t)e^{i\t t}&=&\int_{-\pi}^{\pi}\frac{\dd \theta}{\pi} \int_0^\infty \dd t e^{it(\t+\cos\theta)}\nn\\
&=&\int_{-\pi}^{\pi}\dd \theta\left(\d((\t+\cos\theta))+\pv\left(\frac{i}{(\t+\cos\theta)}\,, \frac1\pi\right)\right)\nn\\
&=&\frac{\chi(|\t|\leq1)}{\sqrt{1-\t^2}}\int_{-\pi}^{\pi}\d(\theta-\arccos \t)+\d(\theta+\arccos \t)\label{eq:TFJ0.delta}\\
&+&\frac{i}{\pi}\pv\int_{\t-1}^{\t+1} \frac{1}{x\sqrt{1-(x-\t)^2}}\label{eq:TFJ0.pv}\,.
\eea
The summand (\ref{eq:TFJ0.delta}) is different from zero only if $|\t|\leq 1$, where it is
$$
2\frac{\chi(|\t|\leq1)}{\sqrt{1-\t^2}}\,,
$$
while (\ref{eq:TFJ0.pv}) gives 
$$
\frac{i}{\pi}\pv\int_{\t-1}^{\t+1} \frac{1}{x\sqrt{1-(x-\t)^2}}=-\frac{\chi(|\t|\leq1)}{\sqrt{1-\t^2}}+i\frac{\chi(|\t|\geq1)}{\sqrt{\t^2-1}}\,.
$$

As for $a<\infty$, comparing (\ref{eq:heav-a}) and (\ref{eq:Heaviside}), we see that we can proceed likewise, simply
 replacing respectively the Dirac  $\de$-function and the principal value with $\sinc$ and $-\cosc$ and (\ref{eq:TFJ0-a}) follows. 
\end{proof}
Of course we can compute the value of any integral of the form
$$
\int_a^b \dd tJ_0(\a t)e^{i\t t}=\frac1\a\int_{\a a}^{\a b} \dd tJ_0(t)e^{i\frac\t\a t}
$$
by straightforward algebraic manipulations. In particular we can obtain the asymptotic behaviour of some integrals useful in this paper. 
To this purpose, we recall the asymptotic expansion of $J_0$ (see \cite{TOI}, 8.451 1)
\bea
J_0(t)&=&\sum_{j\geq 0}\left[J^j_+\frac{e^{it}}{t^{j+1/2}}+J^j_-\frac{e^{-it}}{t^{j+1/2}}\right]\,,\label{eq:AsJ01}\\
J^j_\pm&:=&\frac{(-1)^{[j/2]}}{2^{j+1}\sqrt{\pi}}\frac{\Gamma(j+1/2)}{j!\Gamma(-j+1/2)}\left(1\pm i(-1)^j\right)\label{eq:AsJ02}\,.
\eea
We have the following result.

\begin{lemma}\label{cor:coda-TFJ0}
Let $a>0,\t\in \R$. One has
\be\label{eq:LEMMA-J0-A-INFTY}
\begin{aligned}
\int_a^\infty \dd tJ_0(t)e^{i\t t}&=\frac{1}{\sqrt{a}}\left[\frac{e^{i\t a}}{\sqrt{|1-\t^2|}}\left(\frac{1-\ii}{2}\right)\left( e^{\ii a}\sqrt{|1-\t|}-e^{-\ii a}\sqrt{|1+\t|}\right)+r(\t,a)\right]\\
&+\OO{a\sqrt{1-\t^2}}\,,
\end{aligned}
\ee
where $r(\t,a)$ is uniformly bounded in $\t$ and $a$.
\end{lemma}

\begin{proof}
Since $a>0$ we can use (\ref{eq:AsJ01}), (\ref{eq:AsJ02}), so that we have to evaluate integrals of the form
$$
\int_{a}^{\infty} \dd t\frac{e^{i(\t\pm1) t}}{t^{j+\frac12}}\,,\quad j\geq0\,.
$$
We see immediately that for $j\geq1$ there all these integrals are finite for any $\t$ and one has
$$
\int_{a}^{\infty} \dd t\frac{e^{i(\t\pm1) t}}{t^{j+\frac12}}=\OO{a^{j-\frac12}}\,,\quad\t\mbox{-uniformly}\,.
$$
Thus we set
\be\label{eq:r}
r(\t,a):=\sum_{\stackrel{j\geq 1,\\} {\s=\pm1}}J_\s^j\int_{a}^{\infty} \dd t\frac{e^{i(\t+\s) t}}{t^{j+\frac12}}
\ee
and because of the rapid decay the $J_\s^j$, we can readily check that it is bounded in $\t,a$.

So we are left with
\be
R(a,\t):=\frac{1}{\sqrt{\pi}}\left(\frac{1+\ii}{2}\int_{a}^{\infty} \dd t\frac{e^{\ii(\t+1) t}}{\sqrt{t}}+\frac{1-\ii}{2}\int_{a}^{\infty} \dd t\frac{e^{i(\t-1) t}}{\sqrt{t}}\right)\,.
\ee

To check (\ref{eq:LEMMA-J0-A-INFTY}), we can use the explicit representation of the last integrals in terms of Fresnel functions. Setting $$s_\pm(\t):=\sign(\t+1)$$ and using the asymptotic expansion of $C(x)$ and $S(x)$ (see \cite{TOI}, 8.255), we obtain

\bea
\int_{a}^{\infty} \dd t\frac{e^{\ii(\t\pm1) t}}{\sqrt{t}}&=&\sqrt{\frac{2\pi}{|\t\pm1|}}\left(\frac{1+\ii s(\t)}{2}-C(\sqrt{a|\t\pm1|})-\ii s_\pm(\t) S(\sqrt{a|w\pm1|})\right)\nn\\
&=&\frac{-\ii s_\pm(\t)}{\sqrt{a}}\frac{e^{\ii a(\t\pm1)}}{\sqrt{|\t\pm1|}}+\OO{a\sqrt{|\t\pm1|}}\,.
\eea
Note that for $\t\to\pm1$, $C(\sqrt{a|\t\pm1|}), S(\sqrt{a|\t\pm1|})\to0$ and then the integral diverges as $|\t\pm1|^{-\frac12}$. Thus we have 
$$
R(a,\t)=\frac{e^{i\t a}}{\sqrt{a}\sqrt{|1-\t^2|}}\left(\frac{1-\ii}{2}\right)\left( e^{\ii a}\sqrt{|1-\t|}-e^{-\ii a}\sqrt{|1+\t|}\right)+\OO{a\sqrt{|1-\t^2|}}\,,
$$
whence (\ref{eq:LEMMA-J0-A-INFTY}) follows.
\end{proof}

\begin{lemma}\label{cor:coda-TFJ0}
Let $a>0,\t\in\RRR$. Then
\be\label{eq:J0-SQRT}
\int_0^a \dd t J_0(t) \frac{e^{i\t t}}{\sqrt{a-t}}=\sum_{\s=\pm1}\frac{e^{\ii(\t+\s)a}}{\sqrt{a}}f_\s(\t,a)+\OO{a}\,,
\ee
where $f$ is a bounded function for $\t+\s\neq0$ and such that $\lim_{\t\to\pm1}f_\s(\t,a)=\pi\sqrt a$.
\end{lemma}
\begin{proof}
Choose any $b\in(0,a)$ and split 
\be\label{eq:J0split}
\int_0^a \dd t J_0(t) \frac{e^{i\t t}}{\sqrt{a-t}}=\int_0^b \dd t J_0(t) \frac{e^{i\t t}}{\sqrt{a-t}}+\int_b^a \dd t J_0(t) \frac{e^{i\t t}}{\sqrt{a-t}}\,.
\ee
For the first integral in the r.h.s. we use (\ref{eq:reprJ0}) to get
\bea
\int_0^b \dd t J_0(t) \frac{e^{i\t t}}{\sqrt{a-t}}&=&\int_{-\pi}^\pi\frac{\dd \theta}{\pi}e^{i(\t+\cos\theta) a}\int_{a-b}^a \dd t \frac{e^{-i(\t+\cos\theta) t}}{\sqrt{t}}\nn\\
&=&\frac{1}{\sqrt{a}}\int_{-\pi}^\pi\frac{\dd \theta}{\pi}\frac{1-e^{i(\t+\cos\theta)}}{\t+\cos\theta}+b\OO{a}\,,\nn
\eea
which is finite uniformly in $\t$.

Concerning the second integral in the r.h.s. of (\ref{eq:J0split}), we can use again the asymptotic expansion of $J_0$.  We have to look at terms 
of the form
$$
\int_b^a \frac{\dd t}{\sqrt{a-t}}\frac{e^{i(\t\pm1)t}}{t^{j+\frac12}}\,.
$$
We set:
\bea
u_1(a,\t)&:=&\sum_{\stackrel{j\geq1}{\s=\pm1}}J^j_\s \int_b^a \frac{\dd t}{\sqrt{a-t}}\frac{e^{i(\t+\s)t}}{t^{j+\frac12}}\,,\nn\\
u_2(a,\t)&:=&\sum_{\s=\pm1}J^0_\s \int_b^a \frac{\dd t}{\sqrt{a-t}}\frac{e^{i(\t+\s)t}}{\sqrt t}\nn\,,
\eea
so that 
$$
\int_0^a \dd t J_0(t) \frac{e^{i\t t}}{\sqrt{a-t}}=u_1(a,\t)+u_2(a,\t)+\OO{\sqrt{a}}\,.
$$ 

It is easily verified that as $j\geq1$ each term is $\OO{\sqrt a}$ uniformly in $\t$. Therefore
$$
\sup_\t |u_1(\t,a)|<\infty\,\,\forall\,a>0\,,\quad \lim_{a\to\infty}\sqrt {a} |u_1(\t,a)|<\infty\,\,\forall\, \t\in\RRR\,.
$$
For $j=0$ we have a different behaviour determined by $\t$: as $\t\neq\pm1$ the oscillations still give a decay as $a^{-\frac12}$, but as $\t=\pm1$ we get a finite contribution as $a\to\infty$ and the explicit value of the integral is $\pi$.
This concludes the proof. 
\end{proof}

Finally we compute a combination of Hilbert and Fourier transform applied to $J_0$. 
\begin{lemma}\label{lemma:HFJ}
It holds
\be\label{eq:HFJ}
\matH[\matF[J_0]](\t)=\sqrt{\frac2\pi}\frac{\chi(|\t|>1)}{\sqrt{\t^2-1}}\,.
\ee
\end{lemma}
\begin{proof}
The computation is simple as it uses only (\ref{eq:H-moltiplicatore}) and the parity of $J_0$. We have
\bea
\matH[\matF[J_0]](\t)&=&\int_\R\frac{dt}{\sqrt{2\pi}} J_0(t)\matH[e^{i\t t}]=-i\int_\R\frac{dt}{\sqrt{2\pi}} J_0(t)\sign \t\nn\\
&=&\sqrt{\frac2\pi}\int_0^\infty dt J_0(t)\sin(\t t)=\sqrt{\frac2\pi}\frac{\chi(|\t|>1)}{\sqrt{\t^2-1}}\nn\,.
\eea
As a byproduct of (\ref{eq:HFJ}), since $\matH^2=-1$, we have
\be
\matH\left[\int_0^\infty dtJ_0(t)e^{i\t t}\right](\t)=-i\matH\left[\int_0^\infty dtJ_0(t)e^{-i\t t}\right](\t)\,.
\ee
This concludes the proof.
\end{proof}


\section{Constant magnetic field}\label{B}

Here we analyse the situation in which $h=0$ but $V_0\neq0$. As it is pointed out in \cite{GG}, the existence and uniqueness of the solution for all times is given by the well established theory of time-dependent Schr\"odinger operators (see for instance \cite{nelson}). Here we sketch a direct proof providing explicit rates of convergence to the limiting dynamics. 

Throughout this section $1/\t$ will be always intended in the sense of principal value, even when not explicitly written. We set for brevity $\tilde J(t;\xi):=J_0(gt)e^{i\xi t}$,
\begin{equation}\label{Wbar}
\begin{aligned}
\bar W_{t_0}f:=\int_{t_0}^t\dd t' \tilde J(g(t-t'))&f(\x, t')\,, \qquad \bar Wf:=\int_{-\infty}^t\dd t' \tilde J(g(t-t'))f(\x, t')\,,\\
 &\bar W^c_{t_0}:=\bar W-\bar W_{t_0}\,.
\end{aligned}
\end{equation}
Note that $\bar W$ is a one-side convolution. The asymptotic equation (\ref{duat}) in this case reads
\be\label{eq:equazione.h=0}
(\II + iV_0 \bar W_{t_0})\psi(\xi,t)=1\,.
\ee
Solving this equation is equivalent to solve the following sequence of integral equations (we are simply using the resolvent identity (see e.g. \cite{HP}, sect. 4.8) with the reminder that can be verified to vanish as $t-t_0\to\infty$)
\be\label{eq:neumann}
\left\{
\begin{array}{lll}
\psi_0(\xi,t)&=&(1+i V_0j(\xi,0))^{-1}\,,\\
(\II+i V_0 \bar W)\psi_k(\xi,t)&=&iV_0\bar W^c_{t_0} \psi_{k-1}(\xi,t)\,,\quad k\geq1\,.
\end{array}
\right.
\ee
Therefore we want to prove that for any $k\geq1$ $W^c_{t_0} \psi_{k-1}$ satisfies (\ref{eq:cond-g}) and $\psi_k$ is $\OO{\sqrt{t-t_0}}$. 

A central object of the analysis is
\be\nonumber
j(\xi,\t):=\int_{0}^\infty J_0(g t')e^{i(g\xi +\t) t'}=\frac{\chi(|g\xi+\t|< g)+i\chi(|g\xi+\t|> g)}{\sqrt{g^2-(g\xi+\t)^2}}
\ee
(recall (\ref{eq:TFJ0})). It is worth to stress that this is not the Fourier transform of $\tilde J(t;\xi)$, but rather of $\tilde J(t;\xi)\chi(t>0)$, since
\be\label{eq:relazione-calJ-j}
\calJ(\t,\xi):=\matF[\tilde J(t;\xi)](\t,\xi)=\sqrt{\frac2\pi}\Re [j(\xi,\t)]\,,\quad \calJ(0,\xi)=\sqrt{\frac2\pi}j(\xi,0)\,.
\ee

It is  useful to write it as $t_0\to\infty$ for generic source terms:
\be\label{eq:h=0-cnv}
(\II + iV_0 \bar W)\psi(\xi,t)=g(\xi,t)\,.
\ee
By Fourier transform (\ref{eq:h=0-cnv}) becomes an algebraic relation: 
\be\label{eq:h=0algebraic}
(1+iV_0j(\xi, \t))\matF[\psi]=\matF[g]
\ee
We notice that the operator $\uno+iV_0j(\xi, \t)$ is bounded and it vanishes at $\t=g\xi\pm\sqrt{g^2+V_0^2}$. 
Thus the following result is achieved by simple Fourier inversion
\begin{prop}\label{prop:sol-asin-h=0}
Assume
\be\label{eq:cond-g}
\int |\matF[g](\xi,\t)|<\infty\,,\qquad
\int \frac{|\matF[g](\xi,\t)|}{|(\t-g\xi+\sqrt{g^2+V_0^2})(\t-g\xi-\sqrt{g^2+V_0^2})|}<\infty\,,
\ee
uniformly in $\xi\in[-1,1]$. Then there is a unique $\psi(t;\xi)\in L^2_\xi C_t([-1,1]\times\R)$ solving (\ref{eq:h=0-cnv}). 

If \eqref{eq:cond-g} is not satisfied, but $\matF[g](\xi,\t)=H^>_{t_0}\ast f$ for $f\in L^1(\R)$, where $H^>$ is defined in \eqref{eq:Heaviside},
then there is a unique $\psi(t;\xi)\in L^2_\xi C_t([-1,1]\times(t_0,\infty))$ solving (\ref{eq:h=0-cnv}).
\end{prop}

\begin{rmk}
For any $\psi_0(\xi)\in L^2[-1,1]$ and $\xi_0\in[-1,1]$, the function $g(\xi, t)=e^{ig\xi_0 t}\psi_0(\xi)$ satisfies \eqref{eq:cond-g} 
(with possibly $\xi_0=\xi$), as $\matF[g]=\psi_0(\x)\d(g\xi_0-\t)$, and $|\xi-\xi_0|<1$ while $\sqrt{1+\frac{V_0^2}{g^2}}>1$. In particular the solution to (\ref{duat}) 
for $h=0$ is given by
\be\nonumber
\psi_0(\xi)=\frac{1}{1+i V_0j(\xi,0)}\,.
\ee
\end{rmk}

Now let us denote $\chi^\gtrless_{t_0}(x):=\chi(x\gtrless t_0)$ and 
\be\label{eq:Wcsuf-conv}
(\bar W^c_{t_0} f)(\xi, t):=\int_{t-t_0}^{\infty}\tilde J(\xi,t') f(t-t') dt'=\chi^>_{t_0}(t) (\tilde J\ast (f\chi^<_{t_0}))(t,\xi)\,. 
\ee
Note that
\bea
H^>(\t)&=&\frac12\d(\t)+\frac{i}{2\pi}\pv\left(\frac{1}{\t}\,, \cdot\right)\,,\\
H^<(\t)&=&\frac12\d(\t)-\frac{i}{2\pi}\pv\left(\frac{1}{\t}\,, \cdot\right)\,
\eea
and 
$$
H^\gtrless\_{t_0}(\t)=e^{i\t t_0}H^\gtrless(\t)\,,
$$
where $\de(\t)$ is the Dirac $\de$-function.

We set $f_{t_0}:=e^{-i\t t_0}f$ and use the formula
\be\label{eq_conv-Ht0}
(H_{t_0}^\gtrless\ast f)=\frac12\left(f\pm i e^{i\t t_0}\mathscr{H}[f_{t_0}]\right)\,, 
\ee
where $\mathscr{H}[f]$ is the Hilbert transform of $f$. The following lemma is crucial for the next computations.
\begin{lemma}
We have
\bea\nonumber
\matF[(\bar W^c_{t_0} f)](\xi, \t)&=&\frac14\left[\calJ\matF[f]+e^{i\t t_0}\matH[\calJ\matH[\matF[f]_{t_0}]\right.\nn\\
&+&\left.ie^{i\t t_0}\left(\matH[\calJ\matF[f]_{t_0}]-\calJ\matH[\matF[f]_{t_0}]\right)]\right]\label{eq:FWc}\,.
\eea
Moreover there exist $F,G,H$  linear functionals in $L^1_{\loc}(\R)$, such that for any $g\in L^1_{\loc}(\R)$ the following holds:
\begin{itemize}
\item[i)] for $\matF[f](\t)=g(\t)\d(\t)$,  defining
\bea
G[g](\xi,\t)&:=&ig(0)\sqrt{\frac2\pi}\left(j(\xi,0)-j(\xi,\t)\right)\label{eq:functionalG}\\
H[g](\xi,\t)&:=&(\calJ(\xi,\t)g(\t)-\calJ(\xi,0)g(0))+g(0)(\calJ(\xi,0)-\calJ(\xi,\t)))
\label{eq:functionalH}\,
\eea
and we have
\be\label{eq:FWcdelta}
\matF[(\bar W^c_{t_0} f)](\xi, \t)=\frac{e^{i\t t_0}}{4\t}G[g](\t)+\frac{\d(\t)}{4}H[g](\t)\,.
\ee

\item[ii)] for $\matF[f](\t)=g(\t)\frac{e^{i\t t_0}}{\t}$ we have
$$
\matF[(\bar W^c_{t_0} f)](\xi, \t)=\frac{e^{it_0\t}}{4\t}F[g](\t)
$$
with
\begin{equation}\label{eq:F[g]}
\begin{aligned}
F[g](\xi,\t)&:=\pi\t\d(\t)\calJ(\xi,\t)(\matH[g](\t)-\matH[g](0))+\calJ(\xi,\t)g(\t)-\calJ(\xi,0)g(0) \\
&+\matH[g](0)(\matH[\calJ](0)-\matH[\calJ](\t))+\matH[\calJ\matH[g]](\t)-\matH[\calJ\matH[g]](0)\\
&+i\calJ(\xi,\t)(\matH[g](\t)-\matH[g](0))+i\matH[\calJ g](\t)-\matH[\calJ g](0)\,. 
\end{aligned}
\end{equation}
\end{itemize}
\end{lemma}
\begin{proof}
To obtain (\ref{eq:FWc}) it suffices to note that
$$
\matF[(\bar W^c_{t_0} f)](\xi, \t)=H^>_{t_0}(\t)\ast \left[\calJ(\xi,\t)(H^<_{t_0}\ast \matF[f])(\xi, \t)\right]\,,
$$
and then apply repeatedly (\ref{eq_conv-Ht0}). Then, bearing in mind (\ref{eq:H-delta}) and (\ref{eq:H[1/t]}),
 starting by (\ref{eq:FWc}) a straightforward computation gives
\bea
\matF[(\bar W^c_{t_0} f)](\xi, \t)&=&\frac14\left(\d(\t)(\calJ(\xi,\t)g(\t)-\calJ(\xi,0)g(0))+ig(0)\frac{e^{i\t t_0}}{\t}\left(\calJ(\xi,0)-\calJ(\xi,\t)\right)\right.\nn\\
&-&\left.g(0)e^{i\t t_0}\matH\left[\frac{\calJ(\xi,0)-\calJ(\xi,\t)}{\t}\right]\right)\,. 
\eea
Now we use (\ref{eq:H-ultima-utile}) and Lemma \ref{lemma:HFJ} to get (\ref{eq:FWcdelta}):
\bea
\matF[(\bar W^c_{t_0} f)](\xi, \t)&=&\frac{\d(\t)}{4}(\calJ(\xi,\t)g(\t)-\calJ(\xi,0)g(0))+\frac{\d(\t)}{4}\pi g(0)(\calJ(\xi,0)-\calJ(\xi,\t)))\nn\\
&+&\frac{ig(0)}{4}\frac{e^{i\t t_0}}{\t}\left(\calJ(\xi,0)-\calJ(\xi,\t)\right)-\frac{g(0)}{4}\frac{e^{i\t t_0}}{\t}\left(\matH[\calJ](0)-\matH[\calJ](\t)\right)\nn\\
&=& \frac{\d(\t)}{4}(\calJ(\xi,\t)g(\t)-\calJ(\xi,0)g(0))+\frac{\d(\t)}{4}\pi g(0)(\calJ(\xi,0)-\calJ(\xi,\t)))\nn\\
&+&\frac{ig(0)}{4}\frac{e^{i\t t_0}}{\t}\left(\calJ(\xi,0)-\calJ(\xi,\t)\right)+\frac{g(0)}{4}\frac{e^{i\t t_0}}{\t}\sqrt{\frac2\pi}\frac{\chi(|\t+g\xi|>g)}{\sqrt{(g\xi+\t)^2-g^2}}\,.
\eea
By (\ref{eq:relazione-calJ-j}) the latter line is equal to $H[g](\x,\t)$.

Finally, once again starting from (\ref{eq:FWc}) we obtain the form of $F[g]$ in (\ref{eq:F[g]}) by a direct computation which uses anew  
(\ref{eq:H[1/t]}) and (\ref{eq:H-delta}). 
\end{proof}

For $k\geq1$ let us set $\upsilon_k(\xi, \t):=iV_0\matF[(\bar W^c_{t_0} \psi_{k-1})](\xi, \t)$ so that (\ref{eq:neumann}) becomes
\be\label{eq:psi-upsilon}
\matF[\psi_k](\xi, \t)=\frac{\upsilon_k(\xi, \t)}{1+iV_0j(\xi,\t)}\,. 
\ee
Since $\matF[\psi_0]=\d(\t)\psi_0(\xi)$, a simple iteration of the previous lemma yields
\bea
\u_1(\xi,\t)&=&(iV_0)\frac{e^{i\t t_0}}{\t}\frac{\psi_0(\xi)}{4}G[1](\t)\,,\nn\\
\u_2(\xi,\t)&=& (iV_0)^2\frac{e^{i\t t_0}}{\t}\frac{\psi_0(\xi)}{4}F\left[\frac{G[1]}{1+iV_0j}\right](\t)\,,\nn\\
\u_3(\xi,\t)&=& (iV_0)^3\frac{e^{i\t t_0}}{\t}\frac{\psi_0(\xi)}{4}F\left[\frac{F\left[\frac{G[1]}{1+iV_0j}\right]}{1+iV_0j}(\t)\right](\t)\,,\nn\\
\u_4(\xi,\t)&=& (iV_0)^4\frac{e^{i\t t_0}}{\t}\frac{\psi_0(\xi)}{4}F\left[\frac{F\left[\frac{F\left[\frac{G[1]}{1+iV_0j}\right]}{1+iV_0j}(\t)\right]}{1+iV_0j}\right](\t)\,,\nn\\
&\vdots&\nn
\eea
The nested structure coming out the iterative scheme is clear. Now we wish to invert formula (\ref{eq:psi-upsilon}) to obtain the $\psi_k$.
The next result completes the study of the unperturbed solution.
\begin{lemma}
For any $k\geq1$ one has $\psi_k(\xi,t)=\OO{\sqrt{t-t_0}}$\,. 
\end{lemma}
\begin{proof}
We want to prove the lemma by induction on $k$. For $k=1$ a direct computation yields
\begin{equation}\nonumber
\begin{aligned}
\int_{\RRR}\frac{d\t}{\sqrt{2\pi}}&\frac{\u_1e^{it\t}}{1+iV_0j(\xi,\t)} \\
&=iV_0(1+iV_0j(\xi,0))\frac{\psi(\xi)}{4}\sqrt{\frac2\pi}\int_{\RRR}
 \frac{d\t}{\sqrt{2\pi}} \frac{e^{i\t(t-t_0)}}{\t}\left(\frac{j(\xi,\t)}{1+iV_0j(\xi,\t)}-\frac{j(\xi,0)}{1+iV_0j(\xi,0)}\right)\,. 
\end{aligned}
\end{equation}
Then we have
\be\label{eq:miracolo}
\frac{j(\xi,\t)}{1+iV_0j(\xi,\t)}-\frac{\calJ(\xi,\t)}{1+iV_0\calJ(\xi,\t)}=i\matH\left[\frac{\calJ(\xi,\t)}{1+iV_0\calJ(\xi,\t)}\right]\,.
\ee
The above identity can be proven first noting that
$$
\frac{j(\xi,\t)}{1+iV_0j(\xi,\t)}-\frac{\calJ(\xi,\t)}{1+iV_0\calJ(\xi,\t)}=\frac{j(\xi,\t)-\calJ(\xi,\t)}{1+iV_0(j(\xi,\t)-\calJ(\xi,\t))}\,.
$$
Then writing explicitely 
\bea
\frac{\calJ(\xi,\t)}{1+iV_0\calJ(\xi,\t)}&=&\frac{\chi(|\t+g\xi|<g) }{\sqrt{g^2-(\t+g\xi)^2}+iV_0}\label{eq:denominatore-esplicitoL1}\,,\\
\frac{j(\xi,\t)-\calJ(\xi,\t)}{1+iV_0(j(\xi,\t)-\calJ(\xi,\t))}&=&\frac{\chi(|\t+g\xi|>g) }{\sqrt{g^2-(\t+g\xi)^2}+iV_0}\label{eq:denominatore-esplicito-NOL1}\,,
\eea
we see that (\ref{eq:miracolo}) follows by direct computation using (\ref{eq:denominatore-esplicitoL1}) and (\ref{eq:denominatore-esplicito-NOL1}) (see \cite{king}, vol. 2). Note that (\ref{eq:denominatore-esplicitoL1}) is in $L^1$, but its Hilbert transform is not. Nevertheless one has
$$
\frac{j(\xi,\t)}{1+iV_0j(\xi,\t)}=\left(H^>\ast \frac{\calJ}{1+iV_0\calJ}\right)(\t)\,.
$$

Thus
$$
\matF^{-1}\left[\frac{j(\xi,\t)}{1+iV_0j(\xi,\t)}\right]=\chi(t>0)\matF^{-1}\left[\frac{\calJ(\xi,\t)}{1+iV_0\calJ(\xi,\t)}\right]\,,
$$
and the second factor in the product on the r.h.s. is continuos, so the r.h.s. is continuous in $\R^+$. Moreover
\bea
\psi_1(\xi,t)&=&iV_0(1+iV_0j(\xi,0))\frac{\psi(\xi)}{4}\sqrt{\frac2\pi}\matF^{-1}\left[\frac1\t\left(\frac{j(\xi,\t)}{1+iV_0j(\xi,\t)}-\frac{j(\xi,0)}{1+iV_0j(\xi,0)}\right)\right](t-t_0)\nn\\
&=&\frac{iV_0}{4}\int_0^\infty \frac{d t'}{\sqrt{2\pi}} \sign(t-t_0-t')\matF^{-1}\left[\frac{\calJ(\xi,\t)}{1+iV_0\calJ(\xi,\t)}\right](t')\nn\\
&-&\frac{iV_0}{4}\int_0^\infty \frac{d t'}{\sqrt{2\pi}} \matF^{-1}\left[\frac{\calJ(\xi,\t)}{1+iV_0\calJ(\xi,\t)}\right](t')\nn\\
&=&\frac{iV_0}{2}\int_{t-t_0}^\infty \frac{d t'}{\sqrt{2\pi}} \matF^{-1}\left[\frac{\calJ(\xi,\t)}{1+iV_0\calJ(\xi,\t)}\right](t')\nn
\eea
and the latter integral has a decay of order $\OO{\sqrt{t-t_0}}$ uniformly in $\xi\in[-1,1]$ because of (\ref{eq:denominatore-esplicitoL1}).

From the last chain of equalities we deduce
$$
\matF[\psi_1](\t)=H_{t_0}^<\frac{\calJ(\xi,\t)}{1+iV_0\calJ(\xi,\t)}\,.
$$
Therefore we assume inductively that there are two real constants $a_k,b_k$ (a more precise computation determines $b_k=V_0^k$) such that
$$
\matF[\psi_k](\t)=H_{t_0}^<\frac{\chi(|\t|<a_k)}{\sqrt{a_k^2-\t^2}+ib_k}\,,
$$
which entails $\psi_k(\xi,t)=\OO{\sqrt{t-t_0}}$. We set $\varsigma_k(\xi,\t):=\t e^{-i\t t_0}\matF[\psi_k](\xi,\t)$, so
$$
\u_{k+1}(\xi,\t)= (iV_0)^{k+1}\frac{e^{i\t t_0}}{\t}\frac{\psi_0(\xi)}{4}F[\varsigma_k](\xi,\t)\,.
$$
Using (\ref{eq:F[g]}), after some straightforward manipulations we obtain
\begin{equation}\nonumber
\begin{aligned}
F[\varsigma_k](\t)&=\calJ(\xi,\t)\varsigma_k(\xi,\t)+i\matH[\calJ\varsigma_k](\t)-(\calJ(\xi,\t)\varsigma_k(0)+i\matH[\calJ\varsigma_k](0))\\
&=\left(H^>\ast\calJ\varsigma_k\right)(\t)-\left(H^>\ast\calJ\varsigma_k\right)(0)\,,
\end{aligned}
\end{equation}
with $\calJ(\xi,\t)\varsigma_k(\xi,\t)\in L^1(\R)$ and locally it behaves as $1/\sqrt{\t}$. Hence we are exactly in the same situation as in the case $k=1$ and an analogue computation yields $\psi_{k+1}(\xi,t)=\OO{\sqrt{t-t_0}}$, for any $\xi\in[-1,1]$. 

\end{proof}


\end{document}